\documentclass{egpubl}
\usepackage{pg2026}

\SpecialIssuePaper         % uncomment for final version of Computer Graphics Forum, special issue
\CGFStandardLicense
\usepackage[T1]{fontenc}
\usepackage{dfadobe}  

\usepackage{cite}  % comment out for biblatex with backend=biber
\BibtexOrBiblatex
\electronicVersion
\PrintedOrElectronic
\ifpdf \usepackage[pdftex]{graphicx} \pdfcompresslevel=9
\else \usepackage[dvips]{graphicx} \fi

\usepackage{egweblnk} 
\usepackage{lineno}
\usepackage{booktabs} % For formal tables
\usepackage{verbatim}
\usepackage{subfigure}
\usepackage{amsmath}
\usepackage{enumerate}
\usepackage{wrapfig}
\usepackage{color}
\usepackage{tcolorbox}
\usepackage{alltt}
\usepackage{listings}
\usepackage{hyperref}
\usepackage{algorithm}
\usepackage{algpseudocode}
\usepackage{amssymb}

\newtheorem{lemma}{Lemma}
\newtheorem{theorem}{Theorem}

\title[Projective Affine Body Dynamics for Multibody Systems]%
{Projective Affine Body Dynamics for Multibody Systems}

\author[Zimeng Ye et al.]
{\parbox{\textwidth}{\centering Zimeng Ye$^{1,2}$\orcid{0009-0005-0104-7315}, Xiaowei He\thanks{Corresponding: xiaowei@iscas.ac.cn}$^{1}$\orcid{0000-0002-8870-2482}, Yuzhong Guo$^{1}$\orcid{0009-0000-4578-9701}, Yin Yang$^{3}$\orcid{0000-0001-7645-5931}, Chenfanfu Jiang$^{4}$\orcid{0000-0003-3506-0583} and Hongan Wang$^{1,2}$}
\\
{\parbox{\textwidth}{\centering $^1$Institute of Software, Chinese Academy of Sciences, China\\$^2$University of Chinese Academy of Sciences, China\\$^3$The University of Utah, USA\\$^4$ University of California, Los Angeles, USA}
}
}
\newcommand*{\highlight}{\textcolor{black}}
\begin{document}
	
	% uncomment for using teaser
	\teaser{
		\includegraphics[width=\textwidth]{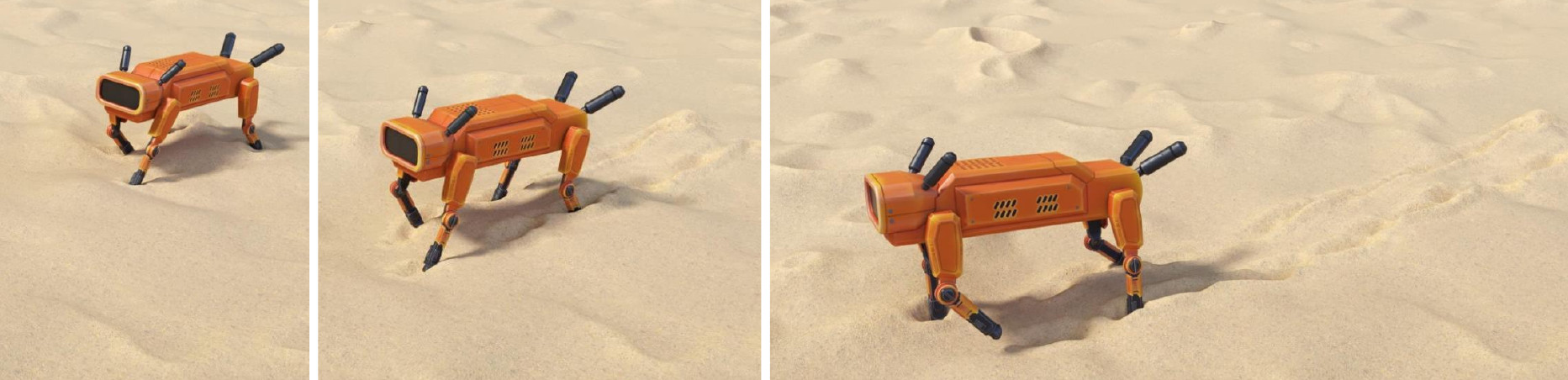}
		\caption{Robotic dog. Our method offers a real-time simulation solution for multibody systems with intricate joints and contact. The simulation of a quadrupedal robotic dog walking on sandy terrain achieves a frame rate of 56 FPS.}
		\label{fig:dog}
	}
	
	%\linenumbers
	\maketitle
	%-------------------------------------------------------------------------
	\begin{abstract}
   Multibody systems have widespread applications in diverse fields such as robotics, entertainment, and animation. 
   Their capability to model and simulate intricate interactions among interconnected bodies makes them invaluable in various domains.
   However, complexity arises with factors like non-smooth contact handling, nonlinearity in flexible joints, and parallelization challenges. 
   We introduce a stable and highly parallel GPU algorithm within affine body dynamics for solving constrained multibody dynamics with nonlinear constraints.
   Our innovation involves reformulating constrained multibody dynamics into a variational form, treating the system as a set of affine bodies connected with peridynamic bonds. 
   This formulation provides a unified model for affine body dynamics, constraints, and contact within the peridynamics framework. 
   It also facilitates the integration of the semi-implicit successive substitution method to solve nonlinear optimization in a global-local iterative manner.
   The proposed method obviates the necessity of assembling a global Hessian, rendering it highly suitable for efficient implementation on GPUs. 
   This allows real-time simulation of intricate interactions involving both rigid and flexible bodies, encompassing contact, joint constraints, and friction.
   \vspace{\baselineskip}  
%-------------------------------------------------------------------------
%  ACM CCS 1998
%  (see https://www.acm.org/publications/computing-classification-system/1998)
% \begin{classification} % according to https://www.acm.org/publications/computing-classification-system/1998
% \CCScat{Computer Graphics}{I.3.3}{Picture/Image Generation}{Line and curve generation}
% \end{classification}
%-------------------------------------------------------------------------
%  ACM CCS 2012
%   (see https://www.acm.org/publications/class-2012)
%The tool at \url{http://dl.acm.org/ccs.cfm} can be used to generate
% CCS codes.
%Example:
\begin{CCSXML}
<ccs2012>
<concept>
<concept_id>10010147.10010371.10010352.10010379</concept_id>
<concept_desc>Computing methodologies~Physical simulation</concept_desc>
<concept_significance>500</concept_significance>
</concept>
</ccs2012>
\end{CCSXML}

\ccsdesc[500]{Computing methodologies~Physical simulation}

\printccsdesc   
\end{abstract}  
	%-------------------------------------------------------------------------
	\section{Introduction}
	
	Multibody systems manifest two distinctive features: system components undergoing finite motions and constraints imposing restrictions on the relative motions.
	Modern tools for multibody dynamics analysis have been developed to address both linear and nonlinear constraints with arbitrary topologies~\cite{Bauchau:2011:Flexible,Deng:2020:Rigid,Probst:2023:Monolithic}.
	Nevertheless, the situation can become significantly complex when considering additional factors, including non-smooth contact handling~\cite{Macklin:2019:NNM,Andrews:2022:Contact}, nonlinearity in flexible bodies~\cite{Deng:2020:Rigid}, and the parallelization of algorithms to leverage modern computing architectures~\cite{Negrut:2014:Parallel,Tasora:2010:GPU}.
	
	The prevailing method for solving constrained multibody dynamics typically begins with formulating the equations of motion with degrees of freedoms (DOFs) for rigid bodies. 
	Subsequently, an additional step involves deriving constraint forces, consolidating the problem into a Linear Complementarity Problem (LCP)~\cite{Andrews:2017:Geometric}.
	Although formulating multibody dynamics as an LCP problem reduces the numerical implementation complexity, it does not adequately address our requirements for modeling complex real-world scenarios.
	For instance, both the contact and joint constraints may demonstrate non-smoothness and nonlinearity.
	To address non-smoothness in contact handling, Chen et al.~\shortcite{Chen:2022:UNB} introduced Incremental Potential Contact (IPC) into multibody dynamics to convert the contact problem into a nonlinear optimization problem; therefore, a unified Newton barrier method can be used to solve complex multibody systems involving contact, friction and articulation constraints.
	However, Newton's method encounters performance bottlenecks when solving problems involving nonlinear constraints.
	Note that as the contacts are actively changing during the simulation, the time-consuming construction of the global Hessian needs to be updated frequently.
	Moreover, the storage pattern of the Hessian matrix becomes unpredictable due to the influence of active contacts on the values of non-zero off-diagonal elements. 
	This complexity poses challenges for achieving an efficient GPU implementation.
	
	This work aims to introduce a GPU-friendly, highly parallel method for addressing the nonlinear optimization problem in constrained multibody dynamics within affine body dynamics. 
	We conceptualize the multibody system as a set of affine bodies subject to rigidity constraints, which allows the use of a projective local/global alternating solver.
	This procedure exhibits similarities to the solver proposed by Li et al.~\shortcite{Li:2019:Fast}. 
	Notably, a significant distinction lies in our formulation of the global energy function, expressed in terms of affine coordinates rather than the reduced degrees of freedom (DOFs) of rigid bodies.
	This distinctive approach enables the modeling of constraints as peridynamic bonds, facilitating the integration of the semi-implicit successive substitution method (SISSM)~\cite{Lu:2023:PPM} to address nonlinear terms in the global step.
	The major contributions of this work include:
	\begin{itemize}
		\item A projective method for solving nonlinear optimization problems that is Hessian-free and relies solely on the first-order derivatives of the objective function;
		\item A novel method to calculate the inverse of the $12 \times 12$ block diagonal matrix, facilitating the development of a highly GPU-compatible algorithm for constrained multibody dynamics;
		\item A method to recover angular velocities of rigid bodies from affine body dynamics by motivating its relationship to both rigid body dynamics and deformable body dynamics.  
	\end{itemize}
	
	%-------------------------------------------------------------------------
	\section{Related Work}
	\subsection{Multibody Dynamics}
	Multibody dynamics is an old subject that focuses on the analysis and simulation of the dynamic behavior of interconnected rigid or flexible bodies.
	While its study can be traced back to the end of the 19th century, it was not modernized until the 1960s when computers were developed~\cite{bremer2008elastic}.
	Addressing the intricacies of highly complex mechanical systems, which encompass numerous interconnected rigid or deformable bodies with joints, contacts, and friction, poses several challenges. 
	These challenges encompass aspects such as system representation~\cite{David:1996:LTD,francu2017unified,gissler2019interlinked}, computational costs~\cite{munawar2020open,Wang:2019:REF}, stability during large time-steps~\cite{Muller:2020:Detailed}, and more~\cite{de2012kinematic}.
	
	In pure rigid body dynamics, early works simply use impulses to prevent interpenetration, yet stability is only guaranteed with some ad hoc treatments, such as velocity thresholding~\cite{Mirtich:1995:Impulse}.
	Guendelman et al.~\shortcite{Guendelman:2003:NRB} introduced an innovative time integration scheme, eliminating the requirement for ad hoc threshold velocities. 
	Additionally, they proposed a novel shock propagation method that enhances the efficiency and visual accuracy in the stacking of objects.
	Kaufman et al.~\shortcite{Kaufman:2005:FFD} proposed to employ a novel contact model that uses mass, location, and velocity information at the moment of maximum compression to reduce the sensitivity of rigid body behaviors to variations in contact configuration.
	Catto~\shortcite{Catto:2005:Iterative} presented a method for caching contact forces that allows contact points to move from step to step.
	This temporal coherence helps amortize the cost of computing accurate contact forces over several frames.
	Erleben~\shortcite{Erleben:2007:VBS} formulated multibody dynamics as a velocity-based complementarity problem and used an iterative projected Gauss-Seidel (PGS) solver to achieve improvements in both performance and quality.
	With the advancement of modern GPUs, Tonge et al.~\shortcite{Tonge:2012:MSJ} expedited the method further by parallelizing each iteration and addressing contacts in blocks. 
	This approach yields wall clock performance comparable to PGS while mitigating visible artifacts.
	Macklin et al.~\shortcite{Macklin:2019:NMM} expanded the solution approach to address nonlinear complementarity problems using a non-smooth Newton iteration. 
	This extension enables the support of nonlinear dynamic models, encompassing hyperelastic deformable bodies and articulated rigid mechanisms.
	
	To guarantee intersection-free configurations at every time step, Ferguson et al.~\shortcite{Ferguson:2021:IFR} borrowed ideas from continuum mechanics and proposed to introduce the incremental potential contact framework to reduced coordinates and rigid body dynamics for time integration and contact response.
	Lan et al.~\shortcite{Lan:2022:ABD} introduced an affine body dynamics (ABD) framework that treats rigid bodies as affine bodies, thus enabling the simulation of both stiff and soft materials within a unified framework.
	With the assistance of ABD, Chen et al.~\shortcite{Chen:2022:UNB} introduced a simulation framework for multibody dynamics employing universal variational integration. 
	This framework facilitates the coupling of rigid and soft bodies interconnected by both linear and nonlinear constraints.
	Recently, Giles et al~\shortcite{Giles:2025:AVBD} use the augmented Lagrangian formulation to enhance vertex block descent, therefore achieving superior performance compared to the state-of-the-art alternatives.
	
	\vspace{-0.15in}
	\subsection{Projective Dynamics}
	The inception of projective dynamics can be traced to position-based dynamics~\cite{Muller:2007:Position,bender2014survey}, where force-based actions are replaced with position-based constraints. 
	While the general position-based approach proves robust in modeling diverse physical phenomena, including elasticity~\cite{Muller:2011:SSO}, contact, friction~\cite{Macklin:2014:UPP}, and fluid incompressibility~\cite{macklin2013position}, its reliability is hindered by sensitivity to numerical parameters such as the simulation time step and iteration count~\cite{Macklin:2016:XPBD}.  

	Liu et al.~\shortcite{Liu:2013:FSM} led the effort to reframe the implicit Euler integration method as an energy minimization problem, solving the linear system globally and addressing all nonlinear terms locally. 
	Bouaziz et al.~\shortcite{Bouaziz:2014:PD} formalized this method by introducing a distance measure to bridge the global and local solvers. 
	Consequently, the global system matrix remains constant and can thus be prefactored during initialization, enabling highly efficient global solves throughout the runtime. 
	\highlight{Overby et al.~\shortcite{Overby:2017:ADMM} further interpreted projective dynamics from the perspective of the alternating direction method of multipliers (ADMM), extending the framework to more general hyperelastic models and dynamic constraints.}

	While projective dynamics inherently exhibits nonlinear characteristics, its convergence behavior closely resembles that of an iterative method employed to solve linear systems. 
	Wang~\shortcite{Wang:2015:CSI} therefore proposed to use the Chebyshev approach to accelerate the convergence rate. 
	While earlier methodologies, such as position-based dynamics or projective dynamics, cater to a restricted range of materials, Liu et al.~\shortcite{Liu:2017:QNM} redefined projective dynamics as a quasi-Newton method.  
	They suggested enhancing the convergence speed by integrating the quasi-Newton method with L-BFGS. 
	He et al.~\shortcite{He:2018:PPM} introduced projective dynamics into peridynamics to model versatile elastoplastic materials with different dimensions. 
	\highlight{To improve the robustness of projective optimization under large rotational motions, Brown and Narain~\shortcite{Brown:2021:WRAPD} proposed a rotation-aware ADMM formulation with adaptive weighting strategies. 
	Recent work has also explored improved treatments of rotational degrees of freedom. 
	Romanyà-Serrasolsas et al.~\shortcite{Romanya:2025:Painless} introduced a differentiable rotation dynamics formulation based on Lie algebra representations, enabling efficient computation of rotational derivatives.}

	Regrettably, the inherent projective nature of the proposed simulator necessitates the definition of the elastic constitutive model in a specific form, and the methodology for modeling general hyperelastic materials remains elusive. 
	In the context of projective dynamics, effectively modeling general hyperelastic materials involves the incorporation of hyperelastic energy functions into the overarching solution.  
	Lu et al.~\shortcite{Lu:2023:PPM} have consequently presented the semi-implicit successive substitution method (SISSM) to address hyperelastic energy functions in the global solve. 
	The adaptability afforded by employing SISSM to incorporate nonlinear terms into the global step not only facilitates a GPU-friendly implementation, but also expands the range of possibilities for projective dynamics in addressing nonlinear optimization problems.
	
	\begin{figure}[t]
		\centering
		\subfigure[]{\includegraphics[width=1.0\linewidth]{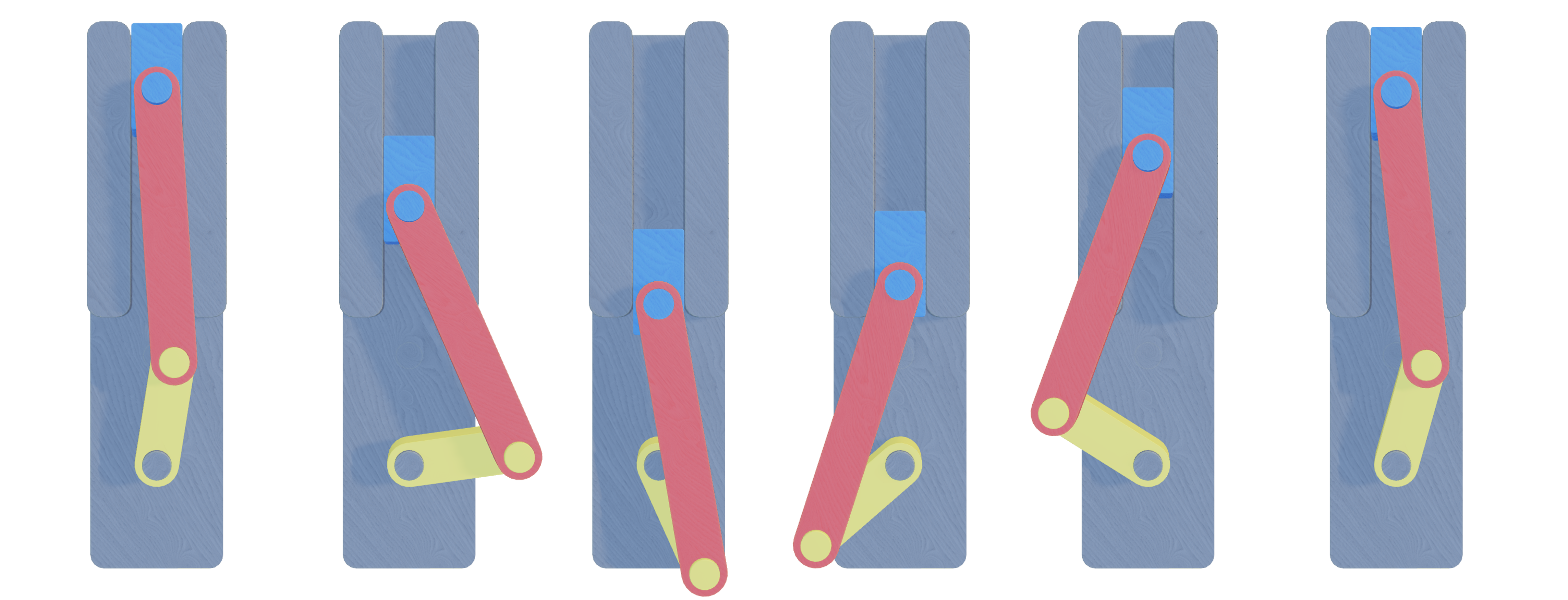}}
		\vspace{-0.15in}
		\subfigure[]{\includegraphics[width=1.0\linewidth]{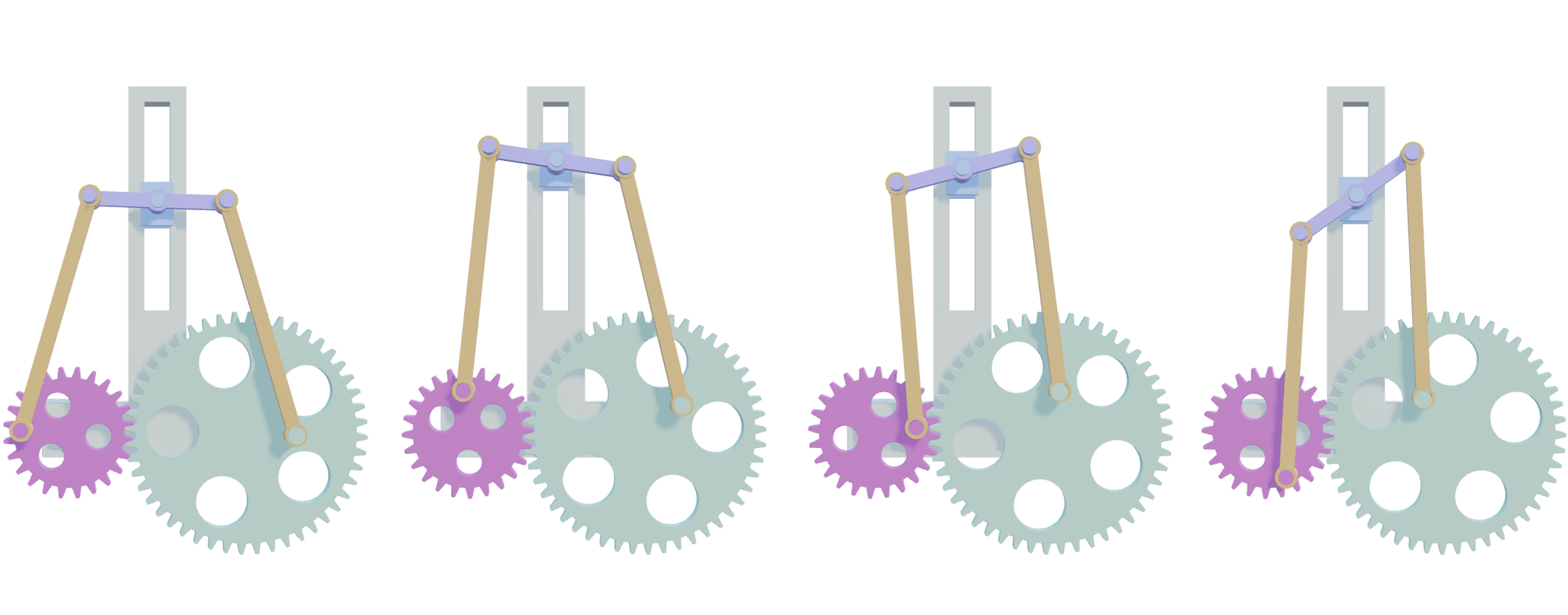}}
		\vspace{-0.08in}
		\caption{Our method is able to simulate complex mechanisms that involve different kinds of joints, contact and friction.}
		\label{fig:linkage}
	\end{figure}
	\vspace{-0.26in}
	\section{Foundations}
	
	%\subsection{Affine Body Dynamics}
	Following the standard ABD by Lan et al.~\shortcite{Lan:2022:ABD}, the kinematics of an affine body is described with a time-varying linear transform $\mathbf{A}\left( t \right) \in {\mathbb{R}^{3 \times 3}}$ and a translation $\mathbf{c}(t) \in {\mathbb{R}^3}$.
	The world coordinates of a material point $\xi$ in body $i$ are given as an affine map
	\begin{equation}
		{\mathbf{y}_{i\xi }} = {\mathbf{A}_i}{\mathbf{x}_{i\xi }} + {\mathbf{c}_i} = \mathbf{J}_{i\xi} \mathbf{q}_i, 
	\end{equation}
	where $\mathbf{x}_{i\xi }$ are the material coordinates defined in body $i$'s frame, $\mathbf{q} = {\left( {{\mathbf{c}^T},\mathbf{a}_1^T,\mathbf{a}_2^T,\mathbf{a}_3^T} \right)^T} \in {\mathbb{R}^{12}}$ with $\mathbf{A} = {[{\mathbf{a}_1},{\mathbf{a}_2},{\mathbf{a}_3}]^T}$ represent the generalized coordinates of body $i$, and ${\mathbf{J}_{i\xi }} = [{\mathbf{I}_3},{\mathbf{I}_3} \otimes {\mathbf{x}_{i\xi }}]$ denotes the Jacobian with $\mathbf{I}_3$ representing the $3 \times 3$ identity matrix.
	To rigidify each affine body, a stiff orthogonality potential
	\begin{equation}
		V_{\bot} \left( \mathbf{q} \right) = \kappa \left\| {\mathbf{A}{\mathbf{A}^T} - \mathbf{I}} \right\|_F^2
	\end{equation}
	is further defined to restore rigid body's rotational DOFs with energy minimization, where $\kappa$ is a constant used to control rigidity.
	
	When each affine body is simulated independently, i.e., no contacts or constraints are considered, the global Hessian for the whole system is simply a $12\times12$ block diagonal matrix, therefore its implementation can be fully parallelized with GPU acceleration.
	However, when contacts or constraints are included, off-diagonal terms will be introduced into the global Hessian, and the unpredictable pattern of the global Hessian can greatly complicate the parallelization as contacts or constraints are actively updated during the simulation.
	
	To alleviate the computational cost, Lan and colleagues~\shortcite{Lan:2022:ABD} proposed to utilize the i-AABBs to speed up the global Hessian assembly.
	However, it is still not possible to estimate and pre-allocate adequate space to store the corresponding $12 \times 12$ off-diagonal blocks due to the sparsity of collision detection.
	For a GPU friendly implementation, it is preferred that the optimal memory cost is a linear function of the number of DOFs and whose upper limit can be estimated at the beginning of simulation.
	Therefore, the objective of this work is to develop a matrix free solver that can model contact, friction, and joints simultaneously and run parallelly on GPU.
	
	\vspace{-0.06in}

	\vspace{-0.06in}
	\section{Projective Affine Body Dynamics}
	To battle above mentioned challenges, let us first revisit the optimization problem for multibody dynamics, then give our solution.
	Table \ref{table:notation} summarizes all notations to facilitate the following discussion.
	\begin{table}[h!]
		\begin{center}
			\caption{Notations.}
			\begin{tabular}{l|l} % <-- Alignments: 1st column left, 2nd middle and 3rd right, with vertical lines in between
				\hline
				\textbf{Notation} & \textbf{Description} \\
				\hline
				$h$ & Time step.      \\
				$\mathbf{q}$ & Generalized coordinates of affine bodies. \\
				$\mathbf{c}$ & Translational coordinates of affine bodies. \\
				$\mathbf{A}$ & Affine matrix of affine bodies \\
				$\mathbf{J}$ & Jacobian of affine bodies.          \\
				$\mathbf{M}$ & Generalized mass matrix of affine bodies. \\
				$\mathbf{H}$ & $12\times12$ coefficient matrix for each affine body.           \\
				$\Pi$ & Affine moment of inertia. \\
				$\mathbf{v}$ & Translational velocity of rigid bodies. \\ 
				$\omega$ & Angular velocity of rigid bodies. \\ 
				$[\mathbf{\omega} ]_ \times$ & \highlight{Skew-symmetric} matrix of $\omega$. \\ 
				$\mathbf{R}$ & Rotation matrix of rigid bodies. \\ 
				$\Psi$ & Potential energy of a bond.\\
				$\varphi$ & Time derivative of $\Psi$ with respect to $d$. \\
				$\mathbf{f}$ & Generalized external force. \\
				$\mathbf{y}$ & World coordinates of contact (constraint) points.  \\
				$\mathbf{x}$ & Material coordinates of contact (constraint) points. \\
				$\mathbf{n}$ & Normal vector along a bond.           \\
				\highlight{$\mathbf{Q}$} & \highlight{Quaternion of rigid bodies.} \\
				\hline
			\end{tabular}
			\label{table:notation}
		\end{center}
	\end{table}
	
	\vspace{-0.06in}
	\subsection{Motivation}
	Assume a multibody system is modeled as a set of $N$ rigid bodies connected by joints to restrict their relative motion.
	The dynamics of a rigid body is uniquely determined by its generalized position $\mathbf{q}$, velocity $\dot{\mathbf{q}}$ and acceleration $\ddot{\mathbf{q}}$.
	Each joint is then defined as functions of the generalized coordinates that can impose from zero to six constraints on the relative motion of a pair of rigid bodies.
	Note if the constraint is a linear function of the generalized coordinates, its time derivative can be uniformly expressed as a product of the Jacobian matrix and the generalized velocity~\cite{chappuis2013constraints}.
	%Hence, solving constrained multibody dynamics necessitates solving only a linear system of equations.
	
	%\begin{figure}[t]
	%  \centering
	%  \includegraphics[width=1.0\linewidth]{../images/primal.pdf}
	%  \caption{(a)Discrete perspective; (b) Continuum perspective. }
	%  \vspace{-0.12in}
	%  \label{fig:duality}
	%\end{figure}
	
	Modeling multibody dynamics becomes challenging when introducing nonlinear constraints, such as incorporating contact potential. 
	Our pivotal approach involves conceptualizing a multibody system as a set of affine bodies connected by peridynamic bonds. 
	These bonds subsequently impose additional forces on each affine body, and the equation of motion for affine bodies could possibly be written as~\cite{Udwadia:1992:New,Lan:2022:ABD}
	\begin{equation}
		\mathbf{M}\ddot{\mathbf{q}} =  - \sum\limits_\xi  {\nabla {\Psi _\xi (\mathbf{q}) }}  + \mathbf{f},
		\label{eq:govern}
	\end{equation}
	where $\mathbf{q} \in \mathbb{R}^{12N}$ represents generalized coordinates of affine bodies, $\mathbf{f} \in \mathbb{R}^{12N}$ represents the generalized external force, and $\Psi _\xi$ represents the bond energy connected to the material point $\xi$.
	Following the update rule of an implicit time integration~\cite{Bouaziz:2014:PD}
	\begin{equation}
		\begin{array}{l}
			{\mathbf{q}^{n + 1}} = {\mathbf{q}^n} + h{{\dot{\mathbf{q}}}^{n + 1}}\\
			{{\dot{\mathbf{q}}}^{n + 1}} = {{\dot{\mathbf{q}}}^n} + h{\mathbf{M}^{ - 1}}\left( { - \sum\nolimits_\xi  {\nabla {\Psi _\xi }\left( {{\mathbf{q}^{n + 1}}} \right)}  + \mathbf{f}} \right)
		\end{array},
		\label{eq:integration}
	\end{equation}
	solving the above equation of motion for affine bodies is equivalent to solving the following optimization problem
	\begin{equation}
		\arg \mathop {\min }\limits_{\mathbf{q}^{n+1}} \frac{1}{{2{h^2}}}\left\| {{\mathbf{M}^{\frac{1}{2}}}\left( {{\mathbf{q}}^{n+1} - {\mathbf{s}^n}} \right)} \right\|_F^2 +  \sum\limits_{\xi} {{\Psi _{\xi}}\left( {{\mathbf{q}}^{n+1}} \right)},
		\label{eq:objective_implicit}
	\end{equation}
	where $\mathbf{s}^n = \mathbf{q}^n + h\dot{\mathbf{q}} + h^2 \mathbf{M}^{-1}\mathbf{f}$.
	Before detailing the solution to the above optimization problem, we first derive several relationships that are fundamental to our method.

	\subsubsection{The relationship to rigid body dynamics}
	%The dependence of $\dot{\mathbf{R}}(t)$ on $\mathbf{R}(t)$ significantly complicates the solution process for rigid body dynamics.
	%Similarly, to accurately capture rotational behavior in affine body dynamics, the current orientation of each affine body must be taken into account during the implicit time integration described in Eq.~\ref{eq:integration}.
	%Since directly incorporating the current rotation matrix on the right-hand side of Eq.~\ref{eq:integration} is not practically feasible, a more efficient approach is to transform all variables from the world frame to the body frame.
	%For instance, by rotating $\mathbf{R}(t)$ into its corresponding body frame, the time derivative of the rotation matrix simplifies to
	%\begin{equation}
	%\dot {\bar{\mathbf{R}}}\left( t \right) = {[\mathbf{\omega} ]_ \times }.
	%\end{equation}
	%Motivated by this observation, we now present a formulation of affine body dynamics within the projective dynamics framework.
	
	\begin{theorem}
		If $\mathbf{f}_{i\xi}$ is the concentrated force imposed on affine body $i$, then Newton's second law in terms of affine coordinates is formulated as
		\begin{equation}
			\mathbf{M}_i \ddot{\mathbf{q}}_i = \sum\limits_{\xi} {\mathbf{J}^T_{i\xi}}\mathbf{f}_{i\xi},
		\end{equation}
		where given the mass density distribution function $\rho$, the generalized mass matrix $\mathbf{M}$ is $12\times12$ block diagonal
		\begin{equation}
			\mathbf{M}_i = \left[ {\begin{array}{*{20}{c}}
					{m_i \mathbf{I}}&0&0&0\\
					0&{\Pi _i}&0&0\\
					0&0&{\Pi _i}&0\\
					0&0&0&{\Pi _i}
			\end{array}} \right],{\kern 10pt}
			\Pi _i  = \int\nolimits_\Omega  {\rho \mathbf{x}_i \mathbf{x}_i^T dV}.
		\end{equation}
	\end{theorem}
	In addition, we define $\Pi _i$ as the \emph{affine moment of inertia}.
	\begin{proof}
		Given concentrated forces $\mathbf{f}_{i\xi}$, we have the force density distribution function $ \mathbf{f}(\mathbf{x}') = \sum\nolimits_{\xi} \mathbf{f}_{i\xi }(\mathbf{x}') = \sum\nolimits_{\xi}{\mathbf{f}_{i\xi }}\delta \left( {\mathbf{x} - \mathbf{x}'} \right)$, where $\delta \left( {\mathbf{x} - \mathbf{x}'} \right)$ is the Dirac delta function.
		By applying Newton's second law, we have 
		\begin{equation}
			\begin{array}{l}
				\begin{aligned}
					\rho \left( {\mathbf{x}'} \right)\ddot{\mathbf{y}}\left( {\mathbf{x}'} \right) &= \mathbf{f}\left( {\mathbf{x}'} \right)\\
					\int_\Omega  {\rho \left( {\mathbf{x}'} \right)\ddot{\mathbf{y}}\left( {\mathbf{x}'} \right)d\Omega '}  &= \int_\Omega  {{\mathbf{J}^T}\left( {\mathbf{x}'} \right)\mathbf{f}\left( {\mathbf{x}'} \right)d\Omega '} \\
					\int_\Omega  {\rho \left( {\mathbf{x}'} \right){\mathbf{J}^T}\left( {\mathbf{x}'} \right)\mathbf{J}\left( {\mathbf{x}'} \right){{\ddot{\mathbf{q}}}_i}d\Omega '}  &= \sum\limits_\xi  {\int_\Omega  {{\mathbf{J}^T}\left( {\mathbf{x}'} \right){\mathbf{f}_{i\xi }}\left( {\mathbf{x}'} \right)d\Omega '} } \\
					{\mathbf{M}_i}{{\ddot{\mathbf{q}}}_i} &= \sum\limits_\xi  {{\mathbf{J}_{i\xi }^T}{\mathbf{f}_{i\xi }}} 
				\end{aligned}
			\end{array}
			\label{eq:derive}
		\end{equation}
		Thus, the formulation for $\mathbf{M}_i$ and $\Pi _i$ can be derived from the above derivation.
	\end{proof}
	
	A deeper understanding of affine body dynamics can be achieved by further examining the properties of the rotation group $\mathrm{SO}(3)$. 
	For any $\mathbf{R}(t) \in \mathrm{SO}(3)$, it is well known that its time derivative is given by
	\begin{equation}
		\dot {\mathbf{R}}\left( t \right) = {[\mathbf{\omega} ]_ \times }\mathbf{R}\left( t \right),{\kern 5pt}{[\mathbf{\omega} ]_ \times } = \left[ {\begin{array}{*{20}{c}}
				0&{ - {\omega _z}}&{{\omega _y}}\\
				{{\omega _z}}&0&{ - {\omega _x}}\\
				{ - {\omega _y}}&{{\omega _x}}&0
		\end{array}} \right]
	\end{equation}
	where $\omega$ is the angular velocity expressed in the world frame, and $[\mathbf{\omega} ]_ \times$ denotes its associated skew-symmetric matrix.
	Note that the time derivative of the rotation matrix depends not only on the angular velocity but also on the current orientation of the rotation matrix itself.
	The following lemma presents a similar relationship between the affine matrix $\mathbf{A}$ and the angular velocity $\omega$.
	\begin{lemma}
		If an affine body is under a rigid transformation, we have $\dot {\mathbf{A}} = {[\mathbf{\omega} ]_ \times } \mathbf{A}$.
		\label{eq:dA}
	\end{lemma}
	\begin{proof}
		According to polar decomposition, the affine matrix has a factorization of the form $\mathbf{A} = \mathbf{R}\mathbf{S}$, where $\mathbf{R}$ is the rotation matrix and $\mathbf{S}$ is a positive semi-definite Hermitian matrix.
		When an affine body undergoes a rotation transformation, the matrix $\mathbf{S}$ remains unchanged, we thus have
		\[\dot {\mathbf{A}} = \dot {\mathbf{R}}{\mathbf{S}} = {[\omega ]_ \times }{\mathbf{R}\mathbf{S}} = {[\omega ]_ \times }\mathbf{A}\]
	\end{proof}
	
	\subsubsection{The relationship to deformable body dynamics}
	\begin{theorem}
		Let ${{\phi}\left( {{\mathbf{x}}} \right)}$ be the energy density function defined over the material space of affine body $i$, then Newton's second law in terms of affine coordinates is formulated as
		\begin{equation}
			{\mathbf{M}_i}{{\ddot{\mathbf{q}}}_i} = - \int_\Omega  {{\nabla _{{\mathbf{q}_i}}}\phi (\mathbf{x}) d\Omega }.
			\label{eq:abdInCon}
		\end{equation}
		
	\end{theorem}
	\begin{proof}
		According to the constitutive relation in continuum mechanics, the elastic force is formulated as $\mathbf{f}(\mathbf{x}) = - \left[ {\frac{{\partial \phi ({\mathbf{x}})}}{{\partial {\mathbf{y}}}}} \right]^T $.
		By taking the following chain rule 
		\[
		- {\nabla _{{\mathbf{q}_i}}}\phi  = {\left[ { - \frac{{\partial \phi }}{{\partial {\mathbf{y}}}} \cdot \frac{{\partial {\mathbf{y}}}}{{\partial {\mathbf{q}_i}}}} \right]^T} = {\mathbf{J}^T}\left( \mathbf{x} \right)\mathbf{f}\left( \mathbf{x} \right),
		\]
		and replace the right hand side of Equation~\ref{eq:derive}, we can verify theorem 4.3.
	\end{proof}

	\begin{lemma}
		If $\Psi _{\xi} (\mathbf{q})$ is the potential energy of a bond connected to affine bodies, we have
		\begin{equation}
			\mathbf{M}_i \ddot{\mathbf{q}}_i = \sum\limits_{\xi}  { - \nabla _{{\mathbf{q}_i}} \Psi _{\xi} ({\mathbf{q}})}
		\end{equation}
	\end{lemma}
	\begin{proof}
		The above lemma can be verified by defining the energy density function as $\phi \left( {\mathbf{x}'} \right) = \sum\nolimits_{\xi} {\Psi _{\xi }}\delta \left( {\mathbf{x} - \mathbf{x}'} \right)$, and substitute $\phi \left( {\mathbf{x}'} \right)$ into Equation~\ref{eq:abdInCon}.
	\end{proof}
	
	\begin{lemma}[transformation of affine moment of inertia]
		After affine transformation, the current affine moment of inertia is updated to ${\Pi}_i = \mathbf{A}_i{\Pi}_0{\mathbf{A}_i^T}$, where $\Pi _0$ represents the initial value calculated in the material space.
	\end{lemma}
	\begin{proof}
		According to the definition of $\Pi_i$, it follows
		\[\begin{array}{l}
			\begin{aligned}
				{\Pi _i} &= \int_\Omega  {\rho {{\mathbf{x}}_i}{\mathbf{x}}_i^Td\Omega} \\
				&= \int_\Omega  {\rho \left( {{\mathbf{A}_i}{{\mathbf{x}}_0}} \right){{\left( {{\mathbf{A}_i}{{\mathbf{x}}_0}} \right)}^T}d\Omega} \\
				&= {\mathbf{A}_i}\left( {\int_\Omega  {\rho {{\mathbf{x}}_0}{\mathbf{x}}_0^Td\Omega} } \right)\mathbf{A}_i^T\\
				&= {\mathbf{A}_i}{\Pi _0}\mathbf{A}_i^T
			\end{aligned}
		\end{array}\]
	\end{proof}

	\subsection{Projective Semi-implicit Solver}
	For simplicity, we uniformly refer to the bond associated with the material point $\xi$ as bond $\xi$ for following discussions.
	Consider an arbitrary bond $\xi$ that connects affine bodies $i$ and $j$. 
	The two contact points in world-space coordinates are given as follows:
	\begin{equation}
		\begin{array}{l}
			\begin{aligned}
				{\mathbf{y}_{i\xi }} &= {\mathbf{A}_i}{\mathbf{x}_{i\xi }} + {\mathbf{c}_i} = \mathbf{J}_{i\xi} \mathbf{q}_i \\
				{\mathbf{y}_{j\xi }} &= {\mathbf{A}_j}{\mathbf{x}_{j\xi }} + {\mathbf{c}_j} = \mathbf{J}_{j\xi} \mathbf{q}_j
			\end{aligned}
		\end{array}.
		\label{eq:twoends}
	\end{equation}
	Assuming that the energy of bond $\xi$ is a function of $d = | \mathbf{y}_{i\xi} - \mathbf{y}_{j\xi} |$, the dynamics of the affine multibody system can be formulated as the following optimization problem:
	\begin{equation}
		\arg \mathop {\min }\limits_\mathbf{q} \frac{1}{{2{h^2}}}\left\| {{\mathbf{M}^{\frac{1}{2}}}\left( {{\mathbf{q}} - {\mathbf{s}^n}} \right)} \right\|_F^2 +  \sum\limits_{\xi} {{\Psi _{\xi}}\left( {d} \right)}, {\kern 5pt}s.t.{\kern 5pt} {V_ \bot}\left( {{\mathbf{q}}} \right) = 0,
		\label{eq:varwithbonds}
	\end{equation}
	where the function ${V_ \bot}\left( {{\mathbf{q}}} \right) = 0$ ensures the rigidity of each affine body.
	Motivated by projected Jacobi or Gauss-Seidel that are commonly used in rigid body dynamics~\cite{Tonge:2012:MSJ}, our approach to solving the constrained optimization problem is divided into two steps.
	In the global step, we temporarily ignore the constraints of ${V_ \bot}\left( {{\mathbf{q}}} \right) = 0$ and apply the semi-implicit successive substitution method~\cite{He:2025:Semi,Lu:2023:PPM} to solve the following nonlinear optimization problem:
	\begin{equation}
		\mathbf{q}^{k + \frac{1}{2}} = \arg \mathop {\min }\limits_\mathbf{q} \frac{1}{{2{h^2}}}\left\| {{\mathbf{M}^{\frac{1}{2}}}\left( {{\mathbf{q}} - {\mathbf{s}^n}} \right)} \right\|_F^2 +  \sum\limits_{\xi} {{\Psi _{\xi}}\left( {d} \right)}.
		\label{eq:varwithbonds}
	\end{equation}
	Then, in the local step, we project $\mathbf{q}^{k + \frac{1}{2}}$ to the nearest feasible point to independently solve the following optimization problem for each affine body
	\begin{equation}
		\mathbf{q}^{k+1} = \arg \mathop {\min }\limits_\mathbf{p} \frac{1}{2}\left\| {{\mathbf{q}^{n + \frac{1}{2}}} - \mathbf{p}} \right\|^2_F,{\kern 5pt}s.t.{\kern 2pt}{V_ \bot }(\mathbf{p}) = 0.
	\end{equation}
	The two steps above are alternately iterated until a user-defined threshold is met or the maximum number of iterations is reached.

	\vspace{-0.18in}
	\subsubsection{Global Solve}
	\label{sec:global_solve}
	Minimizing Eq.~\ref{eq:varwithbonds} is equivalent to finding a solution that satisfies the following equations
	\begin{equation}
		\mathbf{M}_i\left( {\mathbf{q}_i - \mathbf{s}_i^n} \right) = {h^2}\sum\limits_{\xi} \underbrace{\frac{{\highlight{\varphi_\xi} }\left( {{d}} \right)}{d}{\mathbf{J}_{i\xi }^T\left( {{\mathbf{J}_{j\xi}}{\mathbf{q}_{j}} - {\mathbf{J}_{i\xi }}{\mathbf{q}_i}} \right)}}_{ - \frac{{\partial {\Psi _\xi }}}{{\partial {{\mathbf{q}}_i}}}},
		\label{eq:firstorder}
	\end{equation}
	where ${\varphi_\xi }$ represents the derivative of ${\Psi}_\xi$ with respect to $d$, $\mathbf{q}_i$ and $\mathbf{q}_j$ represents the generalized coordinates of affine body $i$ and $j$, respectively.
	Note that ${\varphi_\xi }( {{d}} )$ may include nonlinear terms, which typically necessitate time-consuming Newton-type algorithms with line-search filtering~\cite{Ferguson:2021:IFR} to ensure robust convergence.
	However, constructing the global Hessian at each Newton iteration is computationally expensive on both multi-core CPUs and GPGPUs, particularly when the number of bonds significantly exceeds the number of affine bodies~\cite{Lan:2022:ABD}.
	To address this challenge, we develop a Hessian-free substitution-type algorithm, following the approach proposed by He et al.~\shortcite{He:2025:Semi}.
	The core idea is to separate the potential force for each bond $-{{\partial \Psi _{\xi}} \mathord{\left/ {\vphantom {{\partial \Psi _{\xi} } {\partial \mathbf{y}_{i\xi}}}} \right. \kern-\nulldelimiterspace} {\partial \mathbf{q}_{i}}}$ into the summation of a positive and a negative part
	\begin{equation}
		-\frac{{\partial \Psi_{\xi} }}{{\partial \mathbf{q}_{i}}} = \frac{{{\varphi_{\xi} ^ + }\left( d \right)}}{d}{\mathbf{J}_{i\xi }^T\left( {{\mathbf{J}_{j\xi}}{\mathbf{q}_{j}} - {\mathbf{J}_{i\xi }}{\mathbf{q}_i}} \right)} + \frac{{{\varphi_{\xi} ^ - }\left( d \right)}}{d}{\mathbf{J}_{i\xi }^T\left( {{\mathbf{J}_{j\xi}}{\mathbf{q}_{j}} - {\mathbf{J}_{i\xi }}{\mathbf{q}_i}} \right)},
		\label{eq:separate}
	\end{equation}
	where ${\varphi_{\xi} ^ + }$ indicates $\varphi_{\xi} ^ + \geq 0$ for any $d$ while $\varphi_{\xi} ^ - \leq 0$.
	%Additionally, $\dot{\Psi} ^ +$ signifies the positive portion of $\dot{\Psi}$, i.e., $\dot{\Psi}(d) > 0$ for $d \in (0, \infty)$, while $\dot{\Psi} ^ -$ denotes the negative portion of $\dot{\Psi}$.
	At $k$-th substitution, we treat the positive part implicitly while the negative part explicitly, therefore the potential force can be linearized as
	\begin{equation}
		\begin{array}{l}
			\begin{aligned}
				-\left(\frac{{\partial \Psi_{\xi} }}{{\partial \mathbf{q}_{i}}}\right)^{k+\frac{1}{2}} &= \frac{{{\varphi_{\xi} ^ + }\left( d^{k} \right)}}{d^{k}_\xi}{\mathbf{J}_{i\xi }^T\left( {{\mathbf{J}_{j\xi}}{\mathbf{q}_{j}^{k+\frac{1}{2}}} - {\mathbf{J}_{i\xi }}{\mathbf{q}_i^{k+\frac{1}{2}}}} \right)} \\
				&+ \frac{{{\varphi_{\xi} ^ - }\left( d^k \right)}}{d^{k}}{\mathbf{J}_{i\xi }^T\left( {{\mathbf{J}_{j\xi}}{\mathbf{q}_{j}^{k}} - {\mathbf{J}_{i\xi }}{\mathbf{q}_i^{k}}} \right)}.
			\end{aligned}
		\end{array}
		\label{eq:semi}
	\end{equation}
	Inserting Eq.~\eqref{eq:semi} into the secant equation of the global energy function, we are able to reformulate Eq.~\ref{eq:firstorder} as
	\begin{equation}
		\begin{array}{l}
			\begin{aligned}
				\mathbf{M}_i\left( {\mathbf{q}_i^{k+\frac{1}{2}} - \mathbf{s}_i^n} \right) &= {h^2}\sum\limits_{\xi}{\frac{{\varphi_\xi ^+ }\left( {{d^k }} \right)}{d^k}{\mathbf{J}_{i\xi }^T\left( {{\mathbf{J}_{j\xi}}{\mathbf{q}^{k+\frac{1}{2}}_{j}} - {\mathbf{J}_{i\xi }}{\mathbf{q}^{k+\frac{1}{2}}_i}} \right)}}    \\
				&+ {h^2}\sum\limits_{\xi}{\frac{{\varphi_\xi^- }\left( {{d^k }} \right)}{d^k}{\mathbf{J}_{i\xi }^T\left( {{\mathbf{J}_{j\xi}}{\mathbf{q}^{k}_{j}} - {\mathbf{J}_{i\xi }}{\mathbf{q}^{k}_i}} \right)}}
			\end{aligned}
		\end{array},
		\label{eq:semi_govern}
	\end{equation}
	which yields a square system of $12N$ linear equations.
	Analogous to the Jacobi iteration, by decomposing the coefficient matrix into its diagonal and off-diagonal components, the body-wise update for each affine body $i$ is given by:
	\begin{equation}
		\begin{array}{l}
			\begin{aligned}
				{\mathbf{q}}_i^{k + \frac{1}{2}} &= {\left({\mathbf{H}_i^k}\right)^{ - 1}}\left( {{\mathbf{M}_i\mathbf{s}}_i^n + {h^2}\sum\limits_\xi  {\varphi _\xi ^ - \left( {{d^k}} \right){\mathbf{J}}_{i\xi }^T{\mathbf{n}}_\xi ^k} } \right.\\
				&{\kern 62pt}\left. { +{\kern 2pt} {h^2}\sum\limits_\xi  {\frac{{\varphi _\xi ^ + \left( {{d^k}} \right)}}{d^k}{\mathbf{J}}_{i\xi }^T{{\mathbf{y}}^k_{j\xi }}} } \right),
			\end{aligned}
		\end{array}
		\label{eq:qglobal}
	\end{equation}
	where ${\mathbf{H}_i^k}$ denotes a symmetric, positive-definite $12 \times 12$ matrix as
	\begin{equation}
		{\mathbf{H}_i} = {\mathbf{M}_i} + {h^2}\sum\limits_\xi  {\frac{{\varphi _\xi ^ + \left( {{d}} \right)}}{{{d}}}\mathbf{J}_{i\xi }^T{\mathbf{J}_{i\xi }}},
		\label{eq:H}
	\end{equation}
	${\mathbf{n}}_\xi$ denotes the normal vector pointing from $\mathbf{y}_{i\xi}$ to $\mathbf{y}_{j\xi}$ 
	\begin{equation}
		{\mathbf{n}_\xi } = \frac{{\left( {{\mathbf{J}_{j\xi }}{\mathbf{q}_j} - {\mathbf{J}_{i\xi }}{\mathbf{q}_i}} \right)}}{d}.
	\end{equation}
	It is important to note that Eq.~\ref{eq:qglobal} now conforms to a fixed-point iteration framework, where each update involves only a separable \highlight{$12 \times 12$} matrix for each affine body $i$.
	This property enables our algorithm to be fully compatible with modern GPU architectures, as each affine body can be updated independently.

	\subsubsection{Local Solve}
	In the local step, the purpose is to rigidify each affine body by solving the following constrained optimization problem
	\begin{equation}
		{\mathbf{q}}^{k+1} = \arg \mathop {\min }\limits_\mathbf{p} \frac{1}{2}\left\| {{{\mathbf{q}}_i^{k + \frac{1}{2}}} - \mathbf{p}} \right\|^2_F,{\kern 5pt}s.t.{\kern 2pt}{V_ \bot }(\mathbf{p}) = 0.
	\end{equation}
	Since the translational component of ${{\mathbf{q}}_i^{k + \frac{1}{2}}}$ is independent of the linear transformation and is unconstrained, the translation can be directly \highlight{set} as ${\mathbf{c}}^{k+1}_i = {\mathbf{c}}^{k+\frac{1}{2}}_i$.
	To solve the rotational component, the constrained optimization problem can \highlight{possibly} be reformulated as an unconstrained problem given by:
	\begin{equation}
		Q(\mathbf{A}_i :\kappa) = {\frac{1}{2} {{\left\| {\mathbf{A}_i - \mathbf{A}^{k+\frac{1}{2}}_i} \right\|}_F^2}}  + \kappa \left\| {\mathbf{A}_i{\mathbf{A}^T_i} - \mathbf{I}} \right\|_F^2,
	\end{equation}
	where $\kappa$ can also be regarded as the Lagrange multiplier.
	As the value of $\kappa$ increases, the affine body is expected to become progressively stiffer. However, it is important to note that the optimization problem remains ill-conditioned when $\kappa$ approaches extremely large values to simulate perfectly rigid bodies.
	%\begin{figure}[t]
	%  \centering
	%  \subfigure[World space]{\includegraphics[width=0.9\linewidth]{images/dangling/dangling_bad.pdf}}
	%  \subfigure[Body space]{\includegraphics[width=0.9\linewidth]{images/dangling/dangling.pdf}}
	%  \vspace{-0.12in}
	%  \caption{(a) A direct solve in the world space renders the simulation highly unstable; (b) Within the body space, the simulation is stable and accurate.}
	%  \label{fig:stabilizer}
	%  \vspace{-0.18in}
	%\end{figure}
	Consider, for instance, the special case where $\kappa = \infty$. In this scenario, minimizing $Q(\mathbf{A}_i : \kappa)$ is equivalent to minimizing $\left\| \mathbf{A}_i \mathbf{A}_i^T - \mathbf{I} \right\|_F^2$. However, since $\left\| \mathbf{A}_i \mathbf{A}_i^T - \mathbf{I} \right\|_F^2$ admits an infinite number of local minimizers, identifying a unique minimizer for $Q(\mathbf{A}_i : \kappa)$ becomes infeasible, which may lead to instability in the simulation.
	To address this issue, we apply a polar decomposition to $\mathbf{A}^{k+\frac{1}{2}}_i = \mathbf{R}_i \mathbf{S}_i$, which is a rotation $\mathbf{R}_i$ followed by positive semi-definite symmetric matrix $\mathbf{S}_i$, and directly set $\mathbf{A}^{k+1}_i = \mathbf{R}_i$.
	The resulting solution is subsequently substituted into Eq.~\ref{eq:qglobal} for the next global computation.
	\vspace{-0.24in}
	\subsection{A sparse representation for $\mathbf{H}_i$}
	One issue that remains unsolved is how to efficiently calculate the inverse of the $12 \times 12$ matrix $\mathbf{H}_i$ on GPU for each global step.
	According to GPU specifics, the total number of registers per Streaming Multiprocessor (SM) is usually limited and many modern GPUs have a hard limit of 255 registers per thread~\cite{GPU2026}.
	Therefore, a dense representation for $\mathbf{H}_i$ can not only quickly consume up the limited number of registers for each GPU thread, but also incur a lot of inefficiency in calculating its inverse. 
	For example, directly calculating the inverse of a $12 \times 12$ matrix using LU decomposition may require at least 144 registers.
	
	To address this problem, let us revisit the structure of $\textbf{J}$ in its matrix form
	\begin{equation}\textbf{J} = \left[ {\begin{array}{*{20}{c}}
				1&0&0&{{x}}&{{y}}&{{z}}&0&0&0&0&0&0\\
				0&1&0&0&0&0&{{x}}&{{y}}&{{z}}&0&0&0\\
				0&0&1&0&0&0&0&0&0&{{x}}&{{y}}&{{z}}
		\end{array}} \right].
	\end{equation}
	By exploiting its sparse nature, the $12 \times 12$ matrix $\mathbf{H}_i$ can be simplified to 
	\begin{equation}
		{\mathbf{H}_i} = \left[ {\begin{array}{*{20}{c}}
				\begin{aligned}
					{{a_i}\mathbf{I}}{\kern 5pt}&&{{\mathbf{e}_x}\mathbf{b}_i^T}&&{{\mathbf{e}_y}\mathbf{b}_i^T}&&{{\mathbf{e}_z}\mathbf{b}_i^T}\\
					{{\mathbf{b}_i}\mathbf{e}_x^T}&&{{\mathbf{D}_i}}{\kern 5pt}&&0{\kern 8pt}&&0{\kern 8pt}\\
					{{\mathbf{b}_i}\mathbf{e}_y^T}&&0{\kern 8pt}&&{{\mathbf{D}_i}}{\kern 5pt}&&0{\kern 8pt}\\
					{{\mathbf{b}_i}\mathbf{e}_z^T}&&0{\kern 8pt}&&0{\kern 8pt}&&{{\mathbf{D}_i}{\kern 5pt}}
				\end{aligned}
		\end{array}} \right],
	\end{equation}
	where $\mathbf{e}_x$, $\mathbf{e}_y$ and $\mathbf{e}_z$ are the unit vectors along the x-, y- and z-axis,
	$a_i$, $\mathbf{b}_i$ and $\mathbf{D}_i$ represent a $1\times1$ scalar, a $3\times1$ vector and a $3\times3$ matrix, whose formula are written as follows
	
	\[\begin{array}{l}
		\begin{aligned}
			{a_i} &= {m_i} + {h^2}\sum\limits_\xi  {\frac{{\varphi _{i\xi }^ + \left( d \right)}}{d}}, {\kern 10pt} {\mathbf{b}_i} = {h^2}\sum\limits_\xi  {\frac{{\varphi _{i\xi }^ + \left( d \right)}}{d}{\mathbf{x}_{i\xi }}}, \\
			{\mathbf{D}_i} &= {\Pi _i} + {h^2}\sum\limits_\xi  {\frac{{\varphi _{i\xi }^ + \left( d \right)}}{d}{\mathbf{x}_{i\xi }}\mathbf{x}_{i\xi }^T}.
		\end{aligned}
	\end{array}\]
	Thus, it can be noted that we only need to store $a_i$, $\mathbf{b}_i$ and $\mathbf{D}_i$ rather than the whole matrix, which helps reduce the required register number from 144 to 13.
	More importantly, the above formulation also greatly simplifies the procedure in calculating the inverse of $\mathbf{H}_i$.
	By invoking the following relationship~\cite{Horn:2012:Matrix}
	\[{\left[ {\begin{array}{*{20}{c}}
				\mathbf{I}&\mathbf{B}\\
				{\mathbf{C}}&\mathbf{D}
		\end{array}} \right]^{ - 1}} = \left[ {\begin{array}{*{20}{c}}
			{\mathbf{I} + \mathbf{B}{{\left( {\mathbf{D} - {\mathbf{C}}\mathbf{B}} \right)}^{ - 1}}{\mathbf{C}}}&{ - \mathbf{B}{{\left( {\mathbf{D} - {\mathbf{C}}\mathbf{B}} \right)}^{ - 1}}}\\
			{ - {{\left( {\mathbf{D} - {\mathbf{C}}\mathbf{B}} \right)}^{ - 1}}{\mathbf{C}}}&{{{\left( {\mathbf{D} - {\mathbf{C}}\mathbf{B}} \right)}^{ - 1}}}
	\end{array}} \right]\]
	where $\mathbf{I}$ can be any identity matrix while $\mathbf{B} = \mathbf{C}^T$, and $\mathbf{D}$ represent another square matrix.
	After some simple algebraic operations, the inverse of $\mathbf{H}_i$ has the same structure as $\mathbf{H}_i$, which is written as
	\begin{equation}
		{\kern -5pt}\mathbf{H}_i^{ - 1} = \left[ {\begin{array}{*{20}{c}}
				\begin{aligned}
					{\left( {a_i^{ - 1} + \tilde{a}_i} \right)\mathbf{I}}&&{ - {\mathbf{e}_x}\tilde{\mathbf{b}}_i^T\tilde{\mathbf{D}}_i^{ - 1}}&&{ - {\mathbf{e}_y}\tilde{\mathbf{b}}_i^T\tilde{\mathbf{D}}_i^{ - 1}}&&{ - {\mathbf{e}_z}\tilde{\mathbf{b}}_i^T\tilde{\mathbf{D}}_i^{ - 1}}\\
					{ - \tilde{\mathbf{D}}_i^{ - 1}{{\tilde{\mathbf{b}}}_i}\mathbf{e}_x^T}&&{\tilde{\mathbf{D}}_i^{ - 1}}{\kern 10pt}&&0{\kern 15pt}&&0{\kern 15pt}\\
					{ - \tilde{\mathbf{D}}_i^{ - 1}{{\tilde{\mathbf{b}}}_i}\mathbf{e}_y^T}&&0{\kern 15pt}&&{\tilde{\mathbf{D}}_i^{ - 1}}{\kern 10pt}&&0{\kern 15pt}\\
					{ - \tilde{\mathbf{D}}_i^{ - 1}{{\tilde{\mathbf{b}}}_i}\mathbf{e}_z^T}&&0{\kern 15pt}&&0{\kern 15pt}&&{\tilde{\mathbf{D}}_i^{ - 1}{\kern 10pt}}
				\end{aligned}
		\end{array}} \right],
	\end{equation}
	where $\tilde{a}_i$, $\tilde{\mathbf{b}}_i$ and $\tilde{\mathbf{D}}_i$ are calculated as follows
	\[ \tilde{a}_i =\tilde{\mathbf{b}}_i^T\tilde{\mathbf{D}}_i^{-1}{\tilde{\mathbf{b}}_i}, {\kern 10pt} {\tilde{\mathbf{b}}_i} = \frac{{{\mathbf{b}_i}}}{{{a_i}}},{\kern 10pt}{\tilde{\mathbf{D}}_i} = {\mathbf{D}_i} - \frac{{{\mathbf{b}_i}\mathbf{b}_i^T}}{{{a_i}}}.\]
	Notice the above formulation of $\mathbf{H}_i$ greatly simplifies the inverse calculation.
	Furthermore, it is GPU-friendly due to a reduction of the required number of registers.
	
	\subsection{Update translational and angular velocities of rigid bodies}
	%Similar to updating rigid bodies, we have to handle both the translational and rotational degrees of freedom for affine bodies.
	%According to Lemma~\ref{eq:dA}, the trajectory of an affine body can be updated as
	%\begin{equation}
	%\dot{\mathbf{q}} = {\left( {{\mathbf{v}^T},({[\omega ]_ \times } \mathbf{a}_1)^T,({[\omega ]_ \times } \mathbf{a}_2)^T,({[\omega ]_ \times } \mathbf{a}_3)^T} \right)^T}.
	%\end{equation}
	Given the solution to the optimization problem, updating the translational velocities of rigid bodies is straightforward, as was typically done in position-based dynamics~\cite{Muller:2007:Position}. 
	To compute the angular velocities of rigid bodies, Lemma~\ref{eq:dA} can be applied. 
	Reformulating Lemma~\ref{eq:dA} to $[\omega ]_ \times = \dot{\mathbf{A}} \mathbf{A}^{-1}$ allows the angular velocity components $\omega _x$, $\omega _y$ and $\omega _z$ o be directly derived from the off-diagonal entries of $[\omega ]_ \times$.
	\highlight{To further reduce floating point errors, the final value of $\omega$ is obtained from the average of $[\omega ]_ \times$ and $-[\omega ]^T_ \times$, i.e., $\left( {{{[\omega ]}_ \times } - [\omega ]_ \times ^T} \right)/2$.}
	Algorithm~\ref{eq:alg} summarizes the complete procedure of our method.\highlight{The rotational component of $\mathbf{q}$ is obtained from the quaternion using the standard quaternion-to-matrix conversion~\cite{Shoemake:1985:Quaternion}.}

	%We propose to find the minimizer of $Q(\mathbf{q}_i :\kappa)$ in two passes as shown in Fig.~\ref{fig:twopass}.
	%In the first pass, we keep $\mathbf{p}_i$ as a constant, and a Newton iteration is taken to update the affine matrix only.
	%It is noteworthy that only the inverse of a $9\times9$ Hessian matrix needs to be calculated.
	%Moreover, to mitigate the influence of the initial value of the affine matrix on the convergence rate, we rotate all variables in $Q(\mathbf{q}_i :\kappa)$ into a reference frame. 
	%Subsequently, we update $\mathrm{A}_i$ within the reference frame.
	%To determine $\mathrm{A}_i$ uniquely, we impose an additional condition that $\mathrm{A}^{k+1}_i$ is closest to $\mathrm{A}^{k}_i$ among all affine matrices satisfying the  secant equation of $Q(\mathbf{q}_i :\kappa)$.
	%Besides, we use $\mathbf{y}'_{i\xi}$ instead of $\mathbf{y}_{i\xi}$ to avoid transitional update in this pass.
	%Finally, we rotate $\mathrm{A}_i$ back to the original frame.
	%Then in the second pass, we keep $\mathrm{A}_i$ as a constant, and $\mathbf{p}_i$ is updated with another Newton iteration.
	%Figure~\ref{fig:stabilizer}(bottom) illustrates the accurate rotation and translation achieved with our two-pass strategy.
	
	\begin{algorithm}[t]
		\caption{Projective Affine Body Dynamics}
		\label{alg:pabd}  
		\begin{algorithmic}[1]
			\While {simulating}
			\State Do collision detection;
			\For {affine body $i$}
			\State $\mathbf{v}_i^* \leftarrow \mathbf{v}_i^n + h{\mathbf{f}_{ext}}/{m_i}$;
			\State $\mathbf{c}_i^* \leftarrow \mathbf{c}_i^n + h\mathbf{v}_i^*$;
			\State $\mathbf{Q}_i^* \leftarrow \mathbf{Q}_i^n + \frac{h}{2}\left[ {{\omega _x},{\omega _y},{\omega _z},0} \right]\mathbf{Q}_i^n$;
			\State $\mathbf{q}^*_i \leftarrow (\mathbf{c}_i^*, \mathbf{Q}_i^*)$;
			\State $\mathbf{q}^{k=0}_i \leftarrow \mathbf{q}^*_i$
			\EndFor
			\For {$k < n_{max}$}
			\For {affine body $i$}
			\State $\mathbf{q}_i^{k + \frac{1}{2}} \leftarrow$ solve Equation~\ref{eq:qglobal};       {\kern 13pt}//Global solve
			\State $\mathbf{R}_i \leftarrow {\mathbf{A}}^{k+\frac{1}{2}}_i$;   {\kern 67pt}//Local solve
			\State $\mathbf{q}_i^{k + 1} \leftarrow (\mathbf{c}_i^{k + \frac{1}{2}}, {\mathbf{R}}_i)$;
			\EndFor
			\EndFor
			\For {affine body $i$}
			\State $\mathbf{q}^{n+1}_i \leftarrow \mathbf{q}^{n_{max}}_i$;
			\State $\mathbf{v}_i^{n + 1} = \left( {\mathbf{c}_i^{n + 1} - \mathbf{c}_i^n} \right)/h$;
			\State $\mathbf{\omega} _i^{n + 1} \leftarrow \left( {\mathbf{A}_i^{n + 1} - \mathbf{A}_i^n} \right){\left( {\mathbf{A}_i^n} \right)^{ - 1}}/h$;
			\EndFor
			\EndWhile
		\end{algorithmic}
		\label{eq:alg}
	\end{algorithm} 
	\vspace{-0.06in}
	
	%\begin{algorithm}[t]
	%\caption{Projective Affine Body Dynamics}\label{alg}
	%\While{$t < t_{stop}$}{
		%    Collision detection \;
		%    \While{$k < N$}{
			%        
			%        }
		%}
	%\label{alg:PABD}
	%\end{algorithm}
	
	\section{Results and Discussion}
	All our algorithms involved in this section are implemented using C++ and CUDA, and the execution of all examples is carried out on an NVIDIA RTX 3080 GPU. 
	Collision detection is accelerated with linear BVH~\cite{karras2012maximizing}.
	Coulomb friction is implemented based on~\cite{Andrews:2022:Contact}.
	Table 2 shows statistics for major examples.

	\vspace{-0.06in}
	\subsection{Unit Tests}
	To support multibody dynamics with contact, we follow Macklin et al.~\shortcite{Macklin:2019:NNM} to implement both the soft and hard constraint functions, and then use a different combination of constraints to model different kinds of joints.
	The details on the mathematical description on both the constraint functions and joints are given in the supplementary material for completeness.
	Please also refer to the video to view the demonstration of three basic joints modeled with our method, include the spherical joint, hinge joint and slider joint.
	
	\paragraph{\textbf{Joints}}
	Figure~\ref{fig:linkage} illustrates the method's capability to simulate complex mechanisms involving various types of joints, contact, and friction.
	Figure~\ref{fig:linkage}(top) illustrates a mechanism composed of four rigid bodies that are connected together by three hinge joints and a slider joint.
	Figure~\ref{fig:linkage}(bottom) illustrates a mechanism composed of a driven gear (yellow) and a free gear (grey) that are connected together by a sliding linkage.
	%Figure~\ref{fig:linkage}(c) illustrates a mechanism composed of two gears that are connected by three linkages, both gears are also connected to the fixed structure through hinge joints and the middle linkage is connected to the fixed structure with a slider joint.
	The collision between gears is modeled using a compound shape, with a capsule positioned at the front of each gear.

	\begin{figure}[t]
		\centering
		\includegraphics[width=1.0\linewidth]{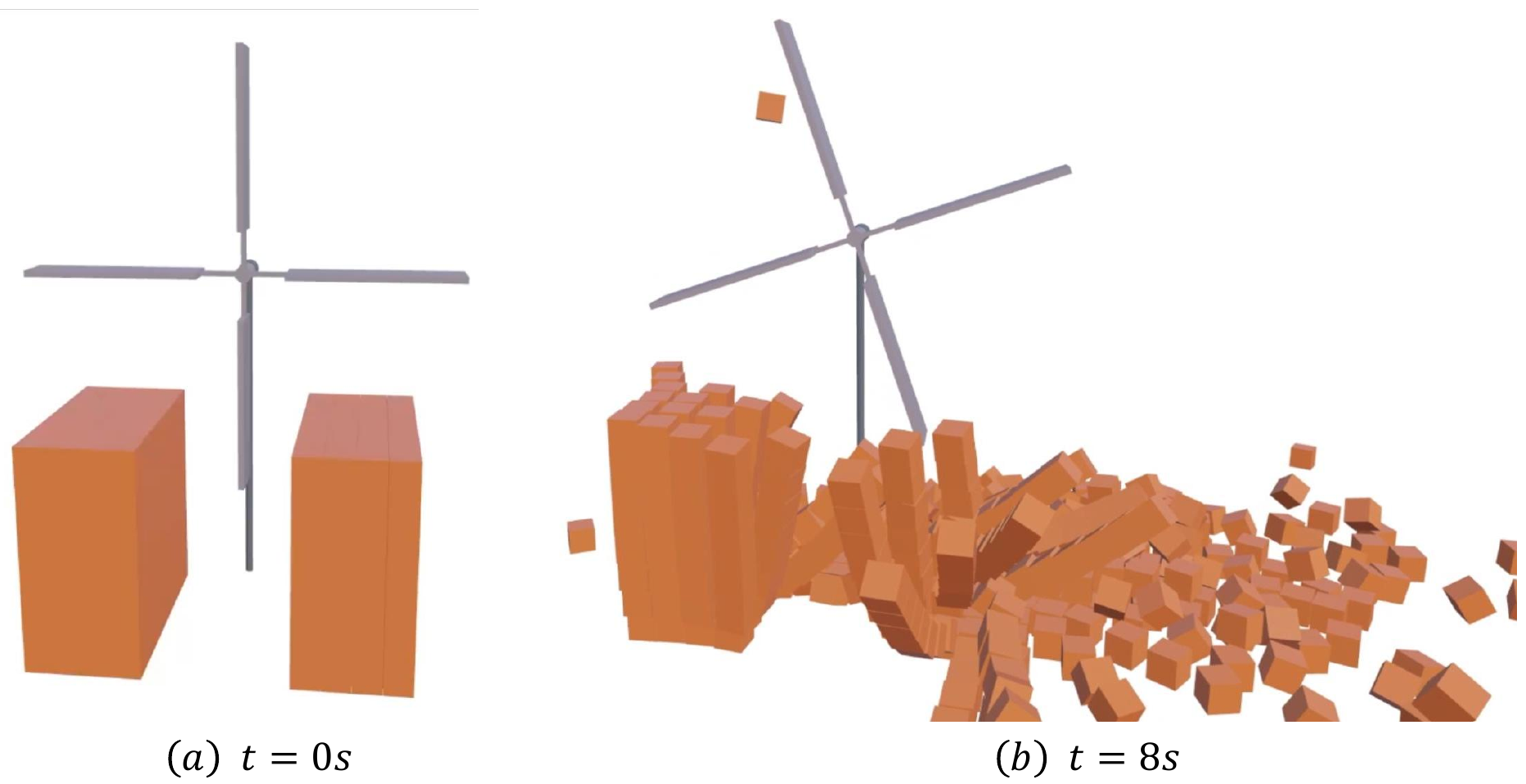}
		\vspace{-0.12in}
		\caption{Simulation of a windmill colliding two piles of boxes. }
		\vspace{-0.12in}
		\label{fig:windmill}
	\end{figure}
	\paragraph{\textbf{Contact and Friction}} Figure~\ref{fig:windmill} \highlight{shows} the ability of our method in simulating a windmill colliding two piles of boxes.
	The windmill is composed of a fixed stand, four wheel blades connected to the fixed stand with driven hinge joint.
	Notice our method can stably solve contact and friction for large number of rigid bodies.
	Compared to the Newton barrier method~\cite{Chen:2022:UNB}, our method achieves a speedup of more than one order.
	\vspace{-0.18in}
	\subsection{Comparison to other methods}
	\subsubsection{\textbf{Comparison to the sequential impulse method}}
	
	%\vspace{-0.06in}
	%\paragraph{Comparison with Bullet} Figure~\ref{fig:bullet} presents a comparison between the simulation of a stack of boxes using Bullet library~\cite{Coumans:2015:Bullet} and our proposed method. 
	%Our approach demonstrates speedups of up to three orders of magnitude, particularly evident when dealing with larger box sizes, attributable to the substantial inherent parallelism in our method.
	
	Figure~\ref{fig:ballhit} presents a comparison between our method and the sequential impulse method~\cite{Catto:2005:Iterative} implemented in PhysX under identical simulation conditions.
	Both methods yield comparable results and accurately capture contact and frictional interactions.
	However, because our method has not yet been fully optimized, it exhibits slightly lower computational performance: it runs at approximately 40 fps, whereas the sequential impulse method achieves approximately 70 fps.
	Despite this, the sequential impulse method encounters difficulties when dealing with extreme mass ratios. 
	As shown in Figure~\ref{fig:largemass}, which depicts a stack of six boxes where each box has ten times the mass of the one below it, the sequential impulse method results in an unstable simulation. 
	In contrast, our method handles this scenario stably.
	
	\begin{figure}[t]
		\centering
		\subfigure[Sequential impulse]{\includegraphics[width=1.0\linewidth]{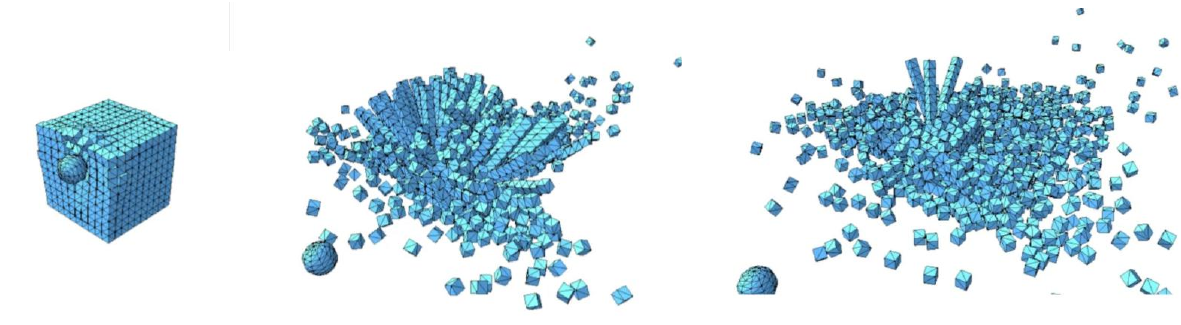}}
		\subfigure[Our method]{\includegraphics[width=1.0\linewidth]{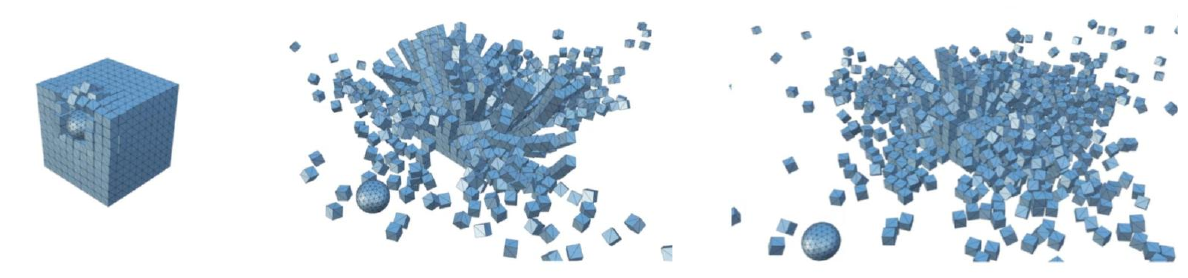}}
		\vspace{-0.12in}
		\caption{A ball hitting a stack of boxes. Our method is able to generate consistent results compared to the sequence impulse method. }
		\label{fig:ballhit}
	\end{figure}
	
	\begin{figure}[t]
		\centering
		\subfigure[Sequential impulse]{\includegraphics[width=0.45\linewidth]{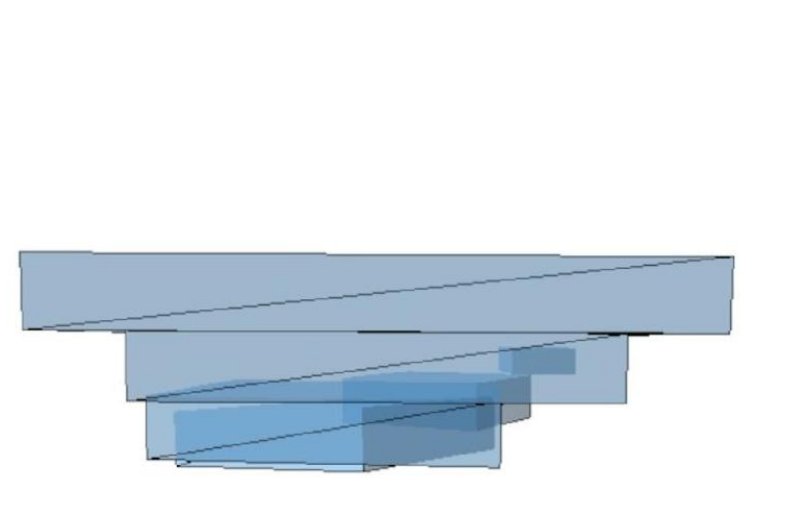}}
		\subfigure[Our method]{\includegraphics[width=0.45\linewidth]{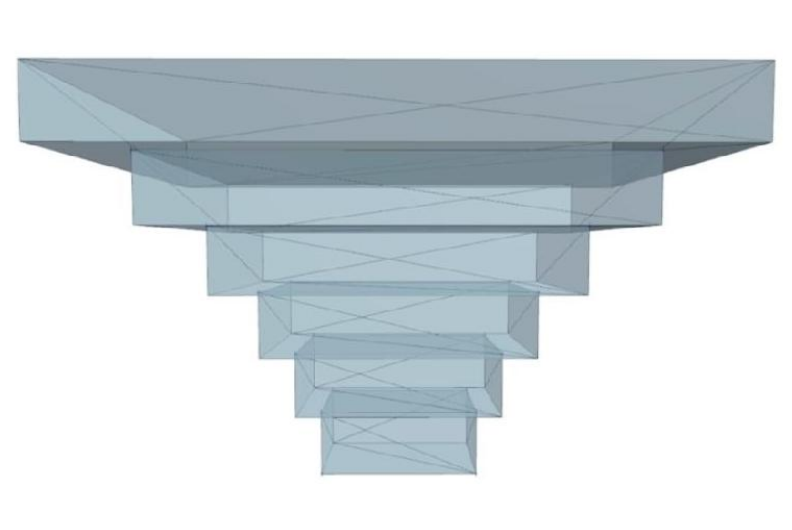}}
		\vspace{-0.12in}
		\caption{Simulation of a stack of boxes where each box has ten times the mass of the one below it. PhysX results in an unstable simulation while our method handles this scenario stably.}
		\vspace{-0.30in}
		\label{fig:largemass}
	\end{figure}
	\vspace{-0.1in}
	\subsubsection{\textbf{Comparison to the AVBD method}}
	Figure~\ref{fig:spring} first demonstrates a simple example using a spring connecting 2 boxes.
	\highlight{To be fair}, both methods apply the following quadratic energy potential
	\begin{equation}
		{\Psi _j} = \frac{1}{2}{\kappa_j}{\left( {{C_j}\left( x \right)} \right)^2},
		\label{eq:quadratic}
	\end{equation}
	where $\kappa_j$ represents the stiffness, $C_j = | d - l |$ describes the elongation of the constraint which has a rest length of $l$.
	Its semi-implicit split of the first order derivative is given in Eq.15 by He et al.~\shortcite{He:2025:Semi}.
	When the spring is initially stretched to two times of the rest length, the convergence rate for the first time step is plotted in Figure~\ref{fig:spring}(a).
	Note that AVBD can obtain the exact solution in a single iteration due to the quadratic convergence of \highlight{Newton's} method. 
	In contrast, the semi-implicit successive substitution method exhibits only linear convergence; consequently, our method requires multiple iterations to achieve a sufficiently accurate solution. \highlight{Figure~\ref{fig:spring}(b) uses the scenario in Fig.~\ref{fig:net}, which is composed of an $11\times11$ grid of boxes connected by 220 springs, where each box is connected to its four nearest neighbors, and all springs are initially stretched toward the center-bottom of the scene. Since the variable stiffness introduced by the augmented Lagrangian strategy in AVBD would increase the total potential energy, we keep all spring stiffnesses in both our method and AVBD identical. As a result, AVBD does not exploit any adaptivity and effectively reduces to standard VBD in this experiment. As can be seen from the figure, our method exhibits slightly faster convergence than VBD during the first few iterations, while both methods eventually converge to similar results.}
	
	Figure~\ref{fig:chain} further presents a comparison between Augmented Projective Affine Body Dynamics (Augmented-PABD) and Augmented Vertex Block Descent (AVBD) in enforcing hard constraints using the following augmented Lagrangian formulation
	\begin{equation}
		{\Psi _j} = \frac{1}{2}{\kappa_j}{\left( {{C_j}\left( x \right)} \right)^2} + \lambda _j {{C_j}\left( x \right)},
		\label{eq:augmented}
	\end{equation}
	where $\lambda$ is the Lagrange multiplier.
	At the beginning of the simulation, we initialize the variables for both methods using $\kappa^0_j = \kappa_{start}$ and $\lambda_j = 0$, where $\kappa_{start} = 10000$.
	In subsequent iterations, the primal and dual variables are updated following the augmented Lagrangian method in AVBD~\cite{Giles:2025:AVBD}
	\begin{equation}
		\begin{array}{l}
			\begin{aligned}
				\kappa_j^{k + 1} &= \kappa_j^k + \beta \left| {C_j^k\left( x \right)} \right|\\
				\lambda _j^{k + 1} &= \kappa_j^k C_j^k\left( x \right) + \lambda _j^k
			\end{aligned}
		\end{array},
	\end{equation}
	where a scaling parameter $\beta$ is set to a constant of $10000$. 
	Meanwhile, we apply the warm starting strategy for both the stiffness and dual variables at every simulation step
	\begin{equation}
		\kappa_j^0 = \max \left( {\gamma {\kern 1pt} \kappa_j^n,{\kappa_{start}}} \right){\kern 10pt}and{\kern 10pt}\lambda _j^0 = \alpha {\kern 1pt} \gamma {\kern 1pt} \lambda _j^t,
	\end{equation}
	where $\alpha = 0.99$ and $\gamma = 0.999$.
	According to the comparison, it can be noticed our method is slightly better in enforcing hard constraints.
	
	Finally, Figure~\ref{fig:avbd}(a) simulates a chain of boxes with a large mass ratio using our method, both with and without the augmented Lagrangian formulation in AVBD~\cite{Giles:2025:AVBD}.
	When the augmented Lagrangian formulation is not applied, the distance constraints are not well preserved due to the large mass of the bottom box.
	However, our method produces results comparable to those of AVBD when the augmented Lagrangian formulation is applied.
	
	\begin{figure}[t]
		\centering
		\subfigure[]{\includegraphics[width=0.44\linewidth]{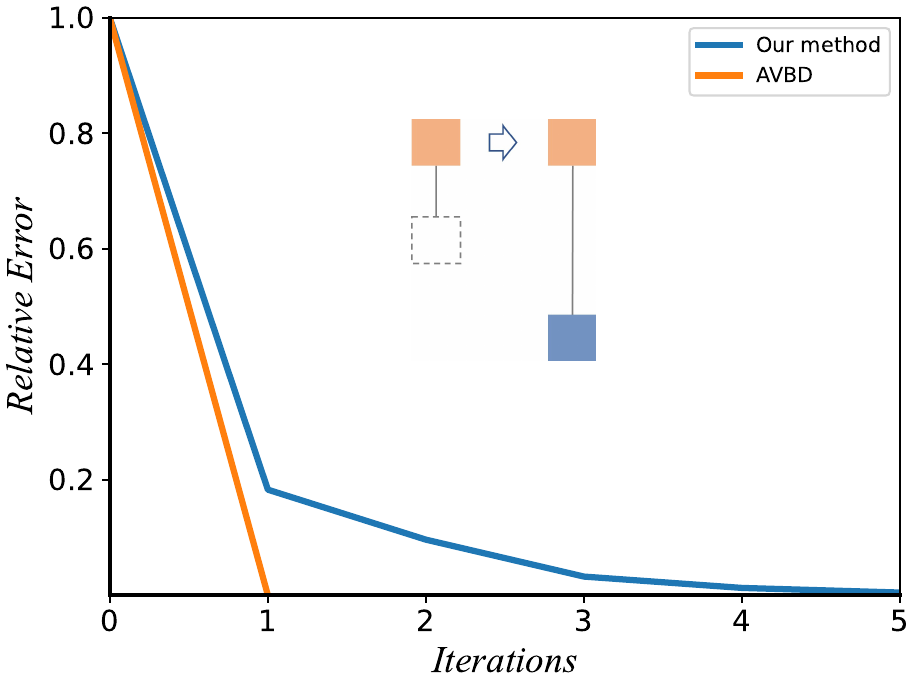}}
		\subfigure[]{\includegraphics[width=0.49\linewidth]{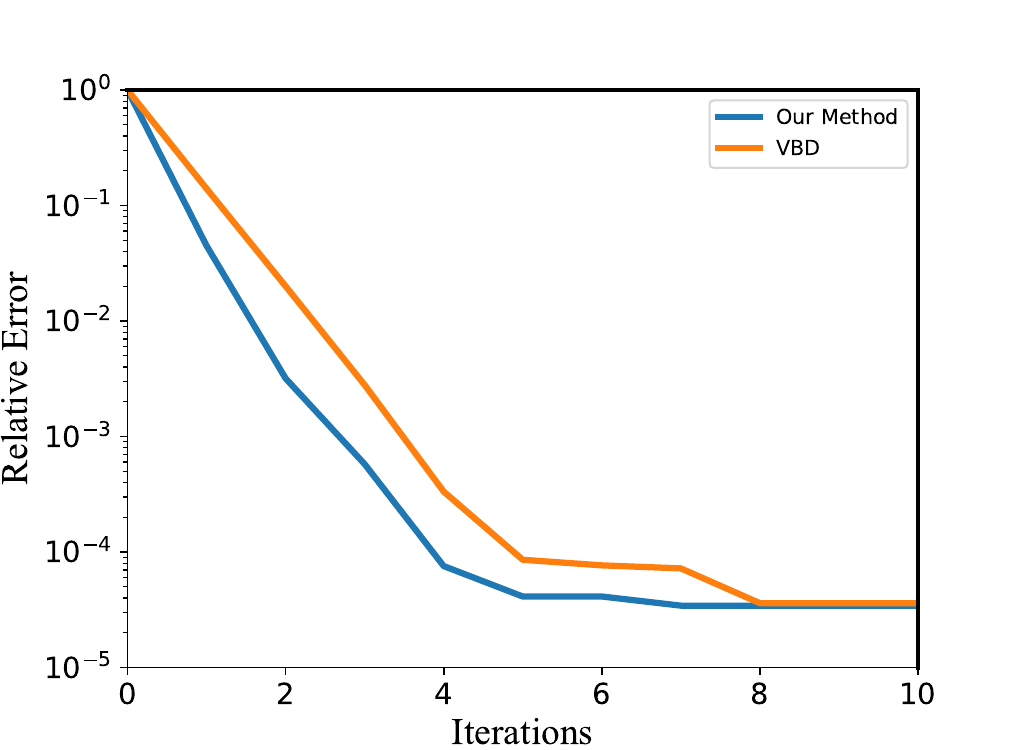}}
		\caption{Comparison of convergence between our method and AVBD. (a) When only one rigid box is simulated, the convergence rate of AVBD is faster than our method. \highlight{(b) In the scenario shown in Fig.~\ref{fig:net}, our method exhibits slightly faster convergence in the early iterations, but the convergence behavior eventually becomes similar. The relative error is defined in the same way as in \cite{Bouaziz:2014:PD}.}}
		\vspace{-0.12in}
		\label{fig:spring}
	\end{figure}

	\begin{figure}[t]
		\centering
		\includegraphics[width=1.0\linewidth]{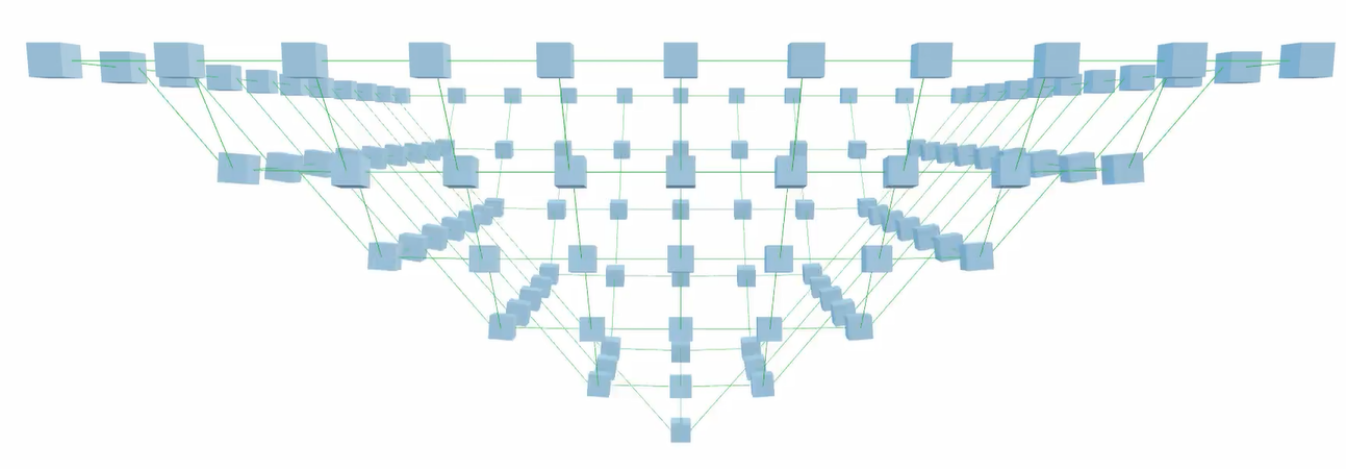}
		\vspace{-0.15in}
		\caption{\highlight{This scenario consists of an $11\times11$ grid of boxes, where each box is connected to its four neighboring boxes by springs, forming 220 spring constraints in total. Initially, the springs connecting each box to its four neighbors are pre-stretched toward the middle-bottom of the scene.}}
		\vspace{-0.15in}
		\label{fig:net}
	\end{figure}
	
	\begin{figure}[t]
		\centering
		\includegraphics[width=1.0\linewidth]{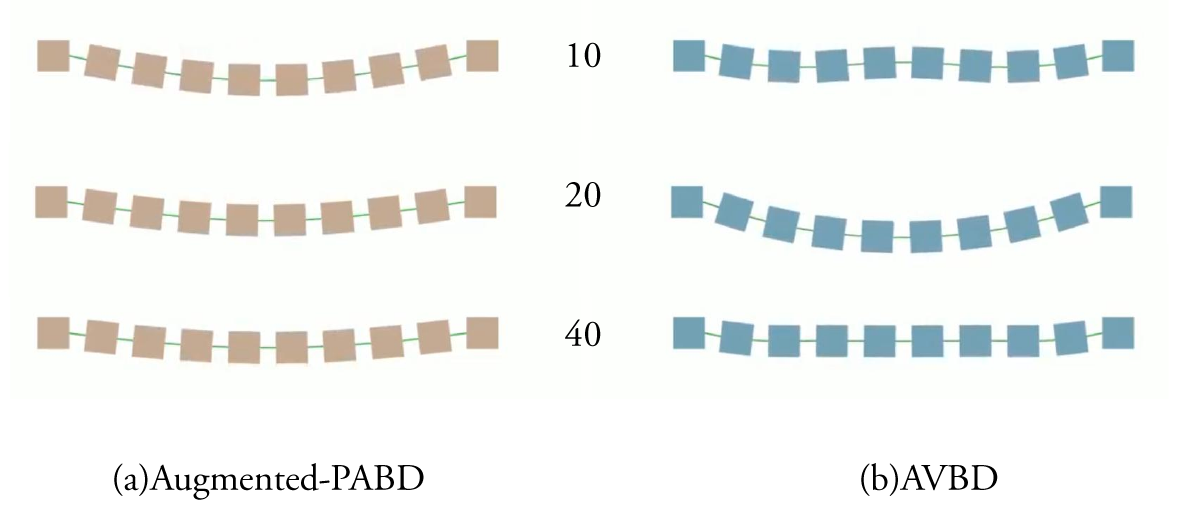}
		\vspace{-0.15in}
		\caption{A chain of boxes. In simulating a chain of rigid boxes with hard constraints, Augmented-PABD performs slightly better than AVBD in settling down the simulation.}
		\vspace{-0.26in}
		\label{fig:chain}
	\end{figure}
	
	\begin{figure}[h]
		\centering
		\includegraphics[width=1.0\linewidth]{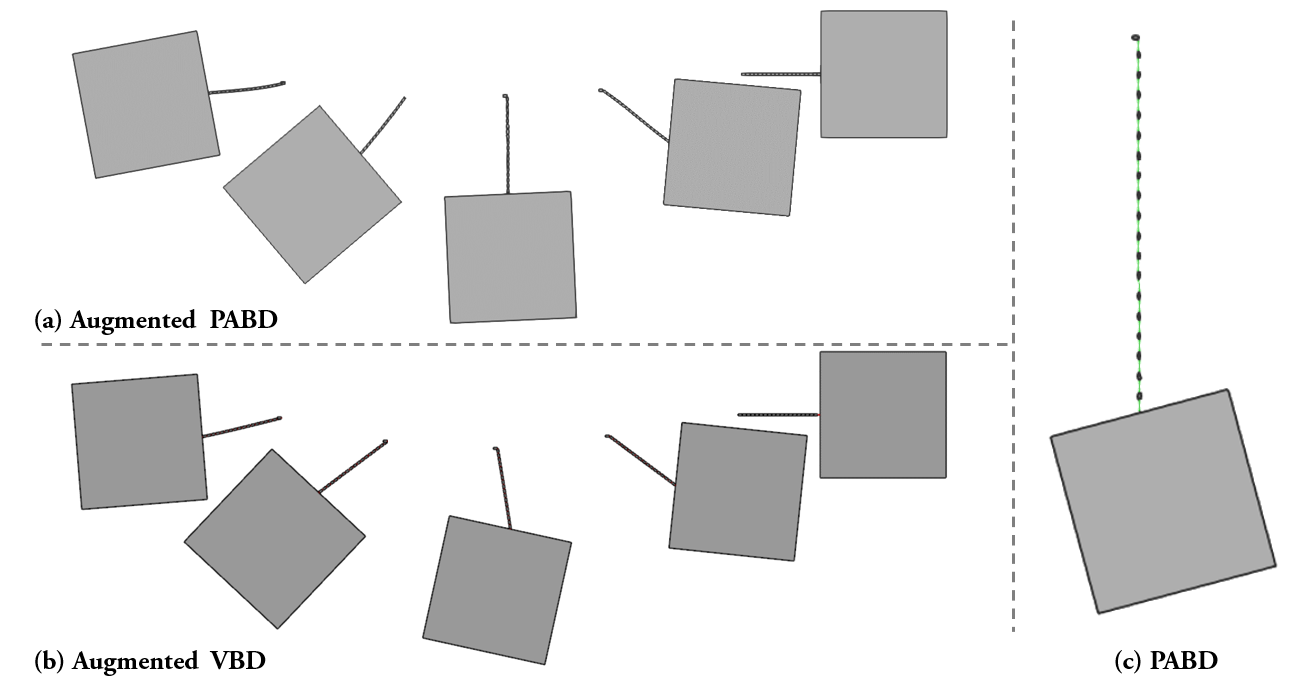}
		\vspace{-0.15in}
		\caption{Comparison between our method and AVBD in simulating a chain of boxes with a large mass ratio. Our method produces results (a) comparable to those of AVBD (\highlight{b}). In contrast, when augmented Lagrangian is not applied, the distance constraints are not well enforced (c). }
		\vspace{-0.12in}
		\label{fig:avbd}
	\end{figure}
	\vspace{-0.20in}
	\subsection{More examples}
	%\paragraph{Complex mechanism} Figure~\ref{fig:linkage} demonstrates a more complicated linkage mechanism by combining different types of joints.
	% Figure~\ref{fig:linkage}(a) illustrates a mechanism composed of four rigid bodies that are connected together by three hinge joints and a slider joint.
	% Figure~\ref{fig:linkage}(b) illustrates a mechanism composed of a driven gear (yellow) and a free gear (grey) that are connected together by a sliding linkage.
	% Figure~\ref{fig:linkage}(c) illustrates a mechanism composed of two gears that are connected by three linkages, both gears are also connected to the fixed structure through hinge joints and the middle linkage is connected to the fixed structure with a slider joint.
	
	%\paragraph{\textbf{Terrain Navigation}} 
	%Figure~\ref{fig:terrain} illustrates a self-propelled car maneuvering across rugged terrain. 
	%The suspension system comprises a suspension body, two front wheels, and two rear wheels. 
	%Each rear wheel is linked via a hinge joint. 
	%Additionally, each front wheel is connected by two hinge joints: one for steering and the other for propelling the car forward. 
	%The application of driving force to the front wheels enables the vehicle to traverse challenging terrain, facilitated by the frictional contact forces between the wheels and the ground.
	
	\subsubsection{\textbf{Bridge}}
	Figure~\ref{fig:bridge} illustrates another jeep using the same suspension to cross a bridge. The bridge is composed of 50 rigid boards, lined together with point connections.
	
	\vspace{-0.06in}
	\subsubsection{\textbf{Coupling with Water}} 
	Figure~\ref{fig:sailboat} depicts a sailboat navigating the ocean, where the ocean is modeled with Gerstner waves~\cite{Tessendorf:2001:Simulating} and the wake behind the sailboat is simulated with the standard shallow water equation~\cite{Garcia:2019:Shallow}.
	The interaction between the boat and the ocean is modeled using a weak coupling strategy, in which the boat's buoyancy is computed directly.
	%The sailboat comprises a rigid body, two deformable sails (174 triangles, 120 vertices), and a rigid horizontal bar for steering one of the sails. 
	%A hinge joint connects the sailboat body to the rigid horizontal bar, controlling the bar's angle. 
	%The sails are linked to the sailboat via spring constraints. 
	%When wind blows, the force applied to the sails is transmitted to the sailboat through constraints, propelling it forward. 
	%Additionally, by manipulating the hinge joint's rotation angle, we can control the sailboat's steering. 
	This example illustrates the robustness of rigid-deformable coupling, demonstrating efficient real-time simulation capabilities.
	
	\vspace{-0.06in}
	\subsubsection{\textbf{Coupling with Sand}} Figure~\ref{fig:dog} shows a real-time simulation of a quadrupedal robotic dog walking on sandy terrain.
	%The quadrupedal robotic dog comprises nine rigid bodies and eight hinge joints. 
	The sandy terrain is simulated using shallow sand equations derived from a height field with a resolution of $256\times 256$~\cite{Zhu:2021:SSE}. 
	%A total of 3.15 million particles are employed to track the surface of the sandy terrain. 
	A straightforward periodic motion is applied to the rotation angles of the hinge joints, enabling the dog to move forward under the influence of external forces.
	The interaction between the dog and sand is similarly modeled using a weak coupling strategy as above.
	Figure~\ref{fig:jeeps} further demonstrates a larger scale of simulation, where a total of 28 jeeps as well as the sand are simulated in real time.
	
	\vspace{-0.30in}
	\highlight{
		\subsubsection{\textbf{Large-scale stacking and collision}} Figure~\ref{fig:largescale} presents balls impacting large-scale 
	$50 \times 50 \times 50$ block piles. 
    Due to the lack of complete performance optimization, the average simulation frame rate is approximately 5 fps.
	}
		
	\begin{table}
		\caption{Time statistics. $N_{body}$ is the number of affine bodies, $N_{joint}$ is the number of joints, $N_{iter}$ is the iteration number, $t_{col}(ms)$ is computational cost for collision detection, $t_{solve}(ms)$ is the computational cost for our projective solver in solving the multibody system only. Computational costs for other materials as well as the coupling are not considered in this table.}
		\label{tab:freq}
		\begin{tabular}{l|c|c|c|c|c}
			\toprule
			Example & $N_{body}$ & $N_{joint}$ & $N_{iter}$ & $t_{col}(ms)$ & $t_{solve}(ms)$ \\
			\midrule
			Figure~\ref{fig:dog} & 9 & 8 & 20 & 2.22 & 2.40\\
			Figure~\ref{fig:linkage}(a) & 4 & 4 & 20 & 1.97 & 2.37\\
			Figure~\ref{fig:linkage}(b) & 8 & 7 & 20 & 1.26 & 2.17\\
			Figure~\ref{fig:windmill} & 515 & 5 & 40 & 1.87 & 4.89\\
			Figure~\ref{fig:ballhit}(b) & 1001 & 0 & 20 & 3.15 & 2.65\\
			Figure~\ref{fig:largemass}(b) & 6 & 0 & 30 & 1.39 & 2.25\\
			Figure~\ref{fig:avbd}(a) & 19 & 18 & 10 & 0 & 1.28\\
			Figure~\ref{fig:avbd}(c) & 19 & 18 & 10 & 0 & 1.25\\
			Figure~\ref{fig:bridge} & 59 & 55 & 20 & 1.73 & 2.05\\
			Figure~\ref{fig:sailboat} & 136 & 10 & 30 & 3.23 & 2.76\\
			Figure~\ref{fig:jeeps} & 140 & 112 & 30 & 0.92 & 2.54\\
			Figure~\ref{fig:largescale}(a) & \highlight{125001} & \highlight{0} & \highlight{20} & \highlight{72.56} & \highlight{87.55}\\
			Figure~\ref{fig:largescale}(b) & \highlight{125001} & \highlight{0} & \highlight{20} & \highlight{74.93} & \highlight{90.24}\\
			\bottomrule
		\end{tabular}
	\end{table}
	
	\begin{figure*}[h]
		\centering
		\includegraphics[width=1.0\linewidth]{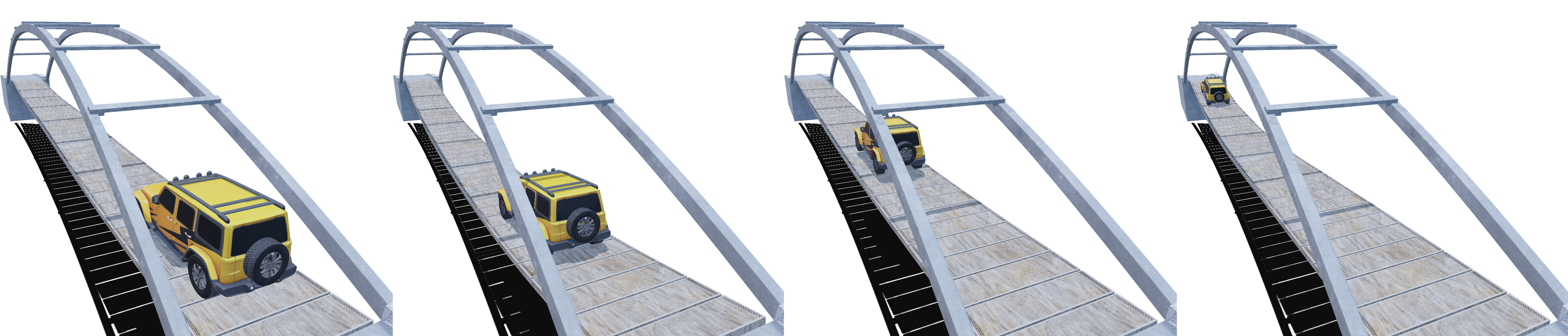}
		\vspace{-0.12in}
		\caption{Bridge. A simulation of a jeep crossing a bridge is conducted, where the bridge comprises 50 rigid boards connected by spherical joints. The example runs at an average frame rate of 55 FPS.}
		\vspace{-0.12in}
		\label{fig:bridge}
	\end{figure*}
	
	\begin{figure*}[h]
		\centering
		\includegraphics[width=1.0\linewidth]{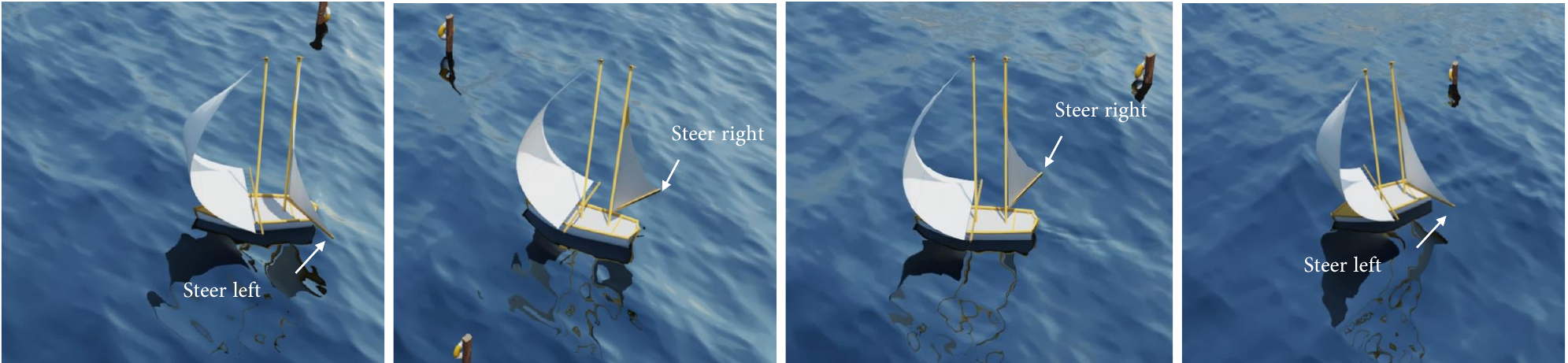}
		\vspace{-0.12in}
		\caption{Sailboat. We simulate a sailboat comprised of a rigid body, two deformable sails, and a rigid horizontal bar connected to the sailboat body with a hinge joint to steer the sailboat. The whole simulate achieves an average rate of 42 FPS.}
		\vspace{-0.24in}
		\label{fig:sailboat}
	\end{figure*}
	\vspace{-0.06in}
	
	\begin{figure*}[h]
		\includegraphics[width=\textwidth]{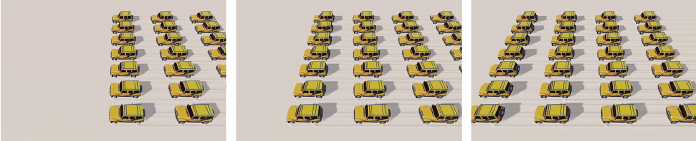}
		\vspace{-0.15in}
		\caption{Wading of jeeps on sand. This scenario is screen-recorded, and the frame rate, accounting for both simulation and rendering, is approximately 40 FPS.}
		\vspace{-0.24in}
		\label{fig:jeeps}
	\end{figure*}

	\begin{figure*}[h]
		\includegraphics[width=\textwidth]{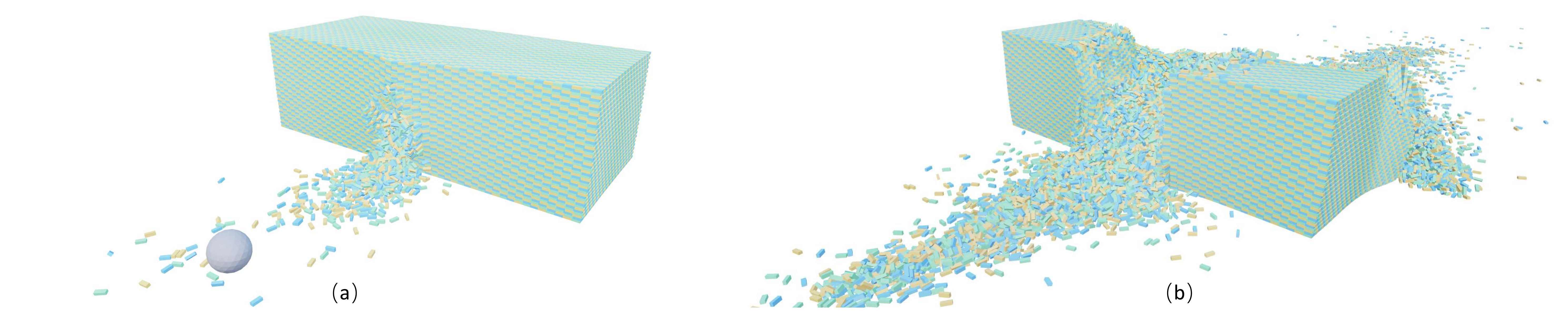}
		\vspace{-0.18in}
		\caption{\highlight{Impact of balls with different masses and velocities on a large-scale $50 \times 50 \times 50$ block pile: (a) smaller mass and lower velocity; (b) larger mass and higher velocity.}}
		\vspace{-0.24in}
		\label{fig:largescale}
	\end{figure*}
	
	\section{Conclusion}
	We have developed projective affine body dynamics, a stable and highly parallel GPU algorithm within affine body dynamics, for solving constrained multibody dynamics with nonlinear constraints.
	Our innovation centers on the reformulation of constrained multibody dynamics into a variational framework, wherein the system is regarded as a set of affine bodies connected by peridynamic bonds. 
	This formulation facilitates the integration of the semi-implicit successive substitution method to address nonlinear terms in the global step and enables a Hessian-free real-time solver for multibody dynamics involving nonlinear joint constraints, contact, and friction.
	
	Our method has several limitations.
	First, it shares the same limitation as the original ABD: simulations can fail under extreme conditions, such as when the rotation between two successive iterations is large.
	Second, in our approach to handling nonlinear equality constraints with polynomial barrier functions, the decomposition of energy into positive and negative components becomes increasingly complex at higher polynomial orders. 
	Finally, how to extend our method to more complex scenarios involving fluids and deformable objects remains unclear, we aim to address these challenges in future work.
	
	\vspace{-0.24in}
    \section*{Acknowledgements}
    We thank the anonymous reviewers for constructive comments and Xukun Luo for demo preparation.    
    This work was supported by the New Generation Artificial Intelligence-National Science and Technology Major Project of China (No. 2025ZD0123902), the National Natural
Science Foundation of China (No.92570206, No.62302490), and the Basic Research Project of ISCAS (No. ISCAS-JCMS-202403).

	\appendix
	%\section{Proof of Theorem}
	
	%-------------------------------------------------------------------------
	% bibtex
	\bibliographystyle{eg-alpha-doi} 
	\bibliography{egbibsample}       
	
	% biblatex with biber
	% \printbibliography                
	
	%-------------------------------------------------------------------------
	
	\clearpage
	
	%\begin{figure*}[h]
	%  \centering
	%  \includegraphics[width=0.505\linewidth]{images/constraint/constraint.pdf}
	%  \includegraphics[width=0.49\linewidth]{images/contact/contact.pdf}
	%  \caption{Left: an increase in the order of the polynomial function results in improved joint constraints; Right: a decrease in the value of $\hat{d}$ leads to a reduction in interpolation magnitude.}
	%  \label{fig:constraintandcontact}
	%\end{figure*}

\end{document}